\def\draft{0}
\def\llncs{0}
\def\anon{0}
\def\noappendix{0}

\ifnum\llncs=1
\documentclass{llncs}
\usepackage{amssymb,url}
\usepackage[noadjust]{cite}
\else
\documentclass[letterpaper]{article}
\usepackage{amsfonts,amsmath,amssymb,amsthm}
\usepackage{fullpage}
\usepackage{palatino}
\fi
\usepackage{physics}
\usepackage{pifont}
\usepackage[english]{babel}
\usepackage[T1]{fontenc}
\usepackage[utf8]{inputenc}

\usepackage{tikz}
\usetikzlibrary{quantikz}

\usepackage[pagebackref,colorlinks=true,linkcolor=magenta,citecolor=blue]{hyperref}
\usepackage[nameinlink]{cleveref}
\usepackage{comment}
\usepackage{lscape}

\ifnum\llncs=1
\title{Quantum Complexity for Discrete Logarithms\\ and Related Problems}
\ifnum\anon=0
\author{%
Minki Hhan\inst{1} \and 
Takashi Yamakawa\inst{2}
Aaram Yun\inst{3}
}
\institute{
KIAS\\
\email{minkihhan@kias.re.kr}
\and
NTT Social Informatics Laboratories\\
\email{takashi.yamakawa@ntt.com}\\
\and
Ewha Womans University\\
\email{aaramyun@ewha.ac.kr}
}
\fi
\else
\theoremstyle{plain}
\title{Quantum Complexity for Discrete Logarithms\\ and Related Problems}
\ifnum\anon=0
\author{
Minki Hhan\thanks{\texttt{E-mail:minkihhan@kias.re.kr}}\\KIAS
\and Takashi Yamakawa\thanks{\texttt{E-mail:takashi.yamakawa.obf@gmail.com}} \\NTT Social Informatics Laboratories
\and Aaram Yun\thanks{\texttt{E-mail:aaramyun@ewha.ac.kr}}\\Ewha Womans University
}
\fi
\fi

\newcommand{\N}{\mathbb N}

\newcommand{\gZ}{\mathbb Z}

\newcommand{\cH}{\mathcal{H}}

\newcommand{\cG}{\mathcal{G}}

\newcommand{\Z}{\mathbb{Z}}

\newcommand{\poly}{\mathop{\operatorname{poly}}}
\newcommand{\bit}{\{0,1\}}

\ifnum\llncs=0
  \theoremstyle{plain}% default
  \newtheorem{theorem}{Theorem}[section]
  \newtheorem{lemma}[theorem]{Lemma}
  
  \newtheorem{proposition}[theorem]{Proposition}
  \theoremstyle{definition}
  \newtheorem{remark}{Remark}
\fi

\newtheorem{clm}{Claim}
\providecommand{\qedhere}{
\ifmmode
  \eqno \def\@badmath{$$}%$$
    \let\eqno\relax \let\leqno\relax \let\veqno\relax
    \hbox{\qed}
\else
  \qed
\fi
}

\newcommand{\red}[1]{\textcolor{red}{#1}}

\newcommand{\QFT}{{\sf{QFT}}}

\newcommand{\adv}{{\sf{Adv}}}

\ifnum\draft=1
\newcommand{\han}[1]{\textcolor{blue}{ $\langle \! \langle$ Minki: ``#1" $\rangle \! \rangle$}}
\newcommand{\takashi}[1]{\textcolor{magenta}{ $\langle \! \langle$ Yamakawa: ``#1" $\rangle \! \rangle$}}
\newcommand{\yun}[1]{\textcolor{red}{ $\langle \! \langle$ Yun: ``#1" $\rangle \! \rangle$}}
\else 
\newcommand{\han}[1]{}
\newcommand{\takashi}[1]{}
\newcommand{\yun}[1]{}
\fi

\begin{document}
\maketitle
\pagestyle{plain}
\ifnum\llncs=0
\begin{abstract}

This paper studies the quantum computational complexity of the discrete logarithm (DL) and related group-theoretic problems in the context of ``generic algorithms''---that is, algorithms that 
do not exploit any properties of the group encoding. 

We establish the quantum generic group model and hybrid classical-quantum generic group model as quantum and hybrid analogs of their classical counterpart. This model counts the number of group operations of the underlying cyclic group $\cG$ as a complexity measure. 
Shor's algorithm for the discrete logarithm problem and related algorithms can be described in this model and make $O(\log |\cG|)$ group operations in their basic form.
We show the quantum complexity lower bounds and (almost) matching algorithms of the discrete logarithm and related problems in these models.
\begin{itemize}
    \item We prove that any quantum DL algorithm in the quantum generic group model must make $\Omega(\log |\mathcal G|)$ depth of group operation queries. This shows that Shor's algorithm that makes $O(\log |\mathcal G|)$ group operations is asymptotically optimal among the generic quantum algorithms, even considering parallel algorithms.  
    \item We observe that some (known) variations of Shor's algorithm can take advantage of classical computations to reduce the number and depth of quantum group operations.
    We show that these variants are optimal among generic hybrid algorithms up to constant multiplicative factors: Any generic hybrid quantum-classical DL algorithm 
    with a total number of (classical or quantum) group operations $Q$ must make $\Omega(\log |\mathcal G|/\log Q)$ quantum group operations of depth $\Omega(\log\log |\mathcal G| - \log\log Q)$.
    \item When the quantum memory can only store $t$ group elements and 
    use quantum random access classical memory (QRACM) of $r$ group elements, 
    any generic hybrid quantum-classical algorithm must make either $\Omega(\sqrt{|\mathcal G|})$ 
    group operation queries in total or $\Omega(\log |\mathcal G|/\log (tr))$ quantum group operation queries. In particular, classical queries cannot reduce the number of quantum queries beyond $\Omega(\log |\mathcal G|/\log (tr))$. 
\end{itemize}

As a side contribution, we show a multiple discrete logarithm problem admits a better algorithm than solving each instance one by one, refuting a strong form of the quantum annoying property suggested in the context of password-authenticated key exchange protocol.
\end{abstract}
\else
\begin{abstract}
This paper studies the quantum computational complexity of the discrete logarithm (DL) and related group-theoretic problems in the context of ``generic algorithms''---that is, algorithms that 
do not exploit any properties of the group encoding. 

We establish the quantum generic group model and hybrid classical-quantum generic group model as quantum and hybrid analogs of their classical counterpart. This model counts the number of group operations of the underlying cyclic group $\cG$ as a complexity measure. 
Shor's algorithm for the discrete logarithm problem and related algorithms can be described in this model and make $O(\log |\cG|)$ group operations in their basic form.
We show the quantum complexity lower bounds and (almost) matching algorithms of the discrete logarithm and related problems in these models.
\begin{itemize}
    \item We prove that any quantum DL algorithm in the quantum generic group model must make $\Omega(\log |\mathcal G|)$ depth of group operation queries. This shows that Shor's algorithm that makes $O(\log |\mathcal G|)$ group operations is asymptotically optimal among the generic quantum algorithms, even considering parallel algorithms.  
    \item We observe that some (known) variations of Shor's algorithm can take advantage of classical computations to reduce the number and depth of quantum group operations.
    We show that these variants are optimal among generic hybrid algorithms up to constant multiplicative factors: Any generic hybrid quantum-classical DL algorithm 
    with a total number of (classical or quantum) group operations $Q$ must make $\Omega(\log |\mathcal G|/\log Q)$ quantum group operations of depth $\Omega(\log\log |\mathcal G| - \log\log Q)$.
    \item When the quantum memory can only store $t$ group elements and 
    use quantum random access classical memory (QRACM) of $r$ group elements, 
    any generic hybrid quantum-classical algorithm must make either $\Omega(\sqrt{|\mathcal G|})$ 
    group operation queries in total or $\Omega(\log |\mathcal G|/\log (tr))$ quantum group operation queries. In particular, classical queries cannot reduce the number of quantum queries beyond $\Omega(\log |\mathcal G|/\log (tr))$. 
\end{itemize}

As a side contribution, we show a multiple discrete logarithm problem admits a better algorithm than solving each instance one by one, refuting a strong form of the quantum annoying property suggested in the context of password-authenticated key exchange protocol.
\end{abstract}
\fi
\newpage
\tableofcontents
\newpage
\section{Introduction}
\ifnum\llncs=1
\setcounter{footnote}{0} 
\fi
The discrete logarithm (DL) problem and related problems have long been fundamental cryptographic primitives in the pre-quantum world~\cite{DH76,Elgamal85}. However, the emergence of quantum computing has drastically altered the landscape of cryptography in the post-quantum world. Shor's algorithm~\cite{Shor94} has demonstrated that the DL problem (and integer factoring)  can be solved in quantum polynomial time, rendering many cryptographic protocols that rely on the DL problem insecure against full-fledged quantum computers. 

\han{Below paragraphs are re-written; contents are almost the same but orders are changed and mentioned Regev here (without detailed explanation).}
Since then, the quantum algorithms solving the DL problem and relevant algorithms have shown significant progress in various directions~\cite{ME98,PG14}. First, hybrid quantum-classical algorithms are suggested, taking advantage of the potentially massive power of classical computation to leverage smaller quantum computers. This direction reduces the required number of sequential quantum group operations by half~\cite{Kaliski17,EH17,Ekera21}. 

The use of parallelism was shown to reduce the depth of quantum DL algorithms significantly~\cite{CW00,HS05}, at the cost of large quantum memory.
The optimization for the arithmetic (or group) operations~\cite{PG14,RNSL17,GE21,HJNRS20} and the other parts~\cite{Kitaev96,ME98} also have been studied extensively. 
The base algorithm of all the above papers is essentially the near variants of Shor's original algorithm, and the complexity of the quantum DL algorithms is still dominated by $O(\log |\cG|)$ group operations for the underlying group $\cG$. 

Very recently, Regev and follow-up research~\cite{Regev23,RV23,EG23} discovered new quantum factoring/DL algorithms that achieve better asymptotic complexity. 
Whether Regev-type algorithms show better practical performance than Shor's is currently unclear and an important open question.\footnote{Later, we will discuss Regev's and related algorithms in more detail.  (Spoiler: both Shor's and Regev's algorithms can be described as generic algorithms, according to our definition. Thus, our group operation lower bound applies, but Regev's better asymptotic performance is \emph{not} about group operation complexity.)}

\yun{I think eventual discussion on Regev's is well-written.  The only possible remaining issue I'm worried is that, while we have very good discussion later, at this point the reader could be confused.  This confusion would be eventually resolved, but the reader has to wait a few more pages.  I'm not sure how to handle this issue gracefully.  Perhaps add a little spoiler around here?}\yun{For example, something like, `Later, we will discuss Regev's and related algorithms in more detail.  (Spoiler: both Shor's and Regev's algorithms are generic, according to our definition, but Regev's better asymptotic performance is \emph{not} about group operation complexity.)'}
% The hybrid quantum-classical algorithms reduce the required number of sequential quantum group operations by half~\cite{Kaliski17,EH17,Ekera21}, and the circuit optimizations are extensively conducted to reduce the circuit size~\cite{GE21,RNSL17,HJNRS20}. Still, their base algorithm is essentially the near variants of Shor's original algorithm, and the complexity of the quantum DL algorithms is still dominated by $O(\log |\cG|)$ group operations for the underlying group $\cG$. 
% As such, no asymptotic improvements have been made since the original algorithm, until Regev and follow-up works~\cite{Regev23,RV23,EG23} recently discovered new types of quantum algorithms.

% While algorithmic optimizations for quantum algorithms solving the DL problem have shown significant progress~\cite{ME98,GE21,PG14,Gid19,RNSL17,HJNRS20}, 
% the base algorithm of these circuit optimizations is essentially the same as Shor's original one or its near variants~\cite{Kaliski17,EH17,Ekera21,Ekera19}. The complexity of the quantum DL algorithms is still dominated by $O(\log |\cG|)$ group operations for the underlying group $\cG$, just like in the original Shor's algorithm. As such, no asymptotic improvements have been made since the original algorithm until very recently.

The lower bounds of the quantum DL algorithms, however, have not been established until now.
This state of affairs raises an important question in the opposite direction:
\begin{center}
\ifnum\llncs=1
    \emph{Is there any quantum lower bound for the DL problem?}
\else
    \emph{Is there any quantum lower bound for the DL problem?}
\fi
\end{center}
The only known relevant result, by Cleve and Watrous~\cite{CW00}, showed a lower bound on the depth for the quantum Fourier transform (QFT), a crucial step of Shor's algorithm. However, there might exist a completely different quantum algorithm that does not rely on the quantum Fourier transform.\footnote{Technically, the phase estimation-based DL algorithm~\cite{Kitaev96} can be done without the QFT (See~\cite{KSV02}). The recent Regev-style algorithms use a different size QFT than Shor's, making the direct implication obscure.} 
To the best of our knowledge, there is no known lower bound, in terms of either time complexity or depth, for the quantum complexity of the DL problem.\footnote{The query complexity lower bound $\Omega(n)$ for the abelian hidden subgroup problem over $(\Z/p\Z)^n$ is proven in~\cite{KNP07}, but it is (almost) independent from $p$ and does not give a meaningful bound for the DL problem.\han{I added this footnote; or we can mention it in related works.}}

\subsection{This Work}
In this paper, we study the hardness of the discrete logarithm problem and related problems by considering a natural class of quantum algorithms referred to as generic algorithms.
A generic quantum algorithm is an algorithm that does not take advantage of the special properties of the encodings of group elements. 
Instead, these algorithms only use group operations only in a black-box manner, potentially in superposition.
Almost all quantum DL algorithms can be described in this model, including Shor's and Regev's (Again, we refer~\Cref{sec:discussion} for a more detailed discussion on the Regev's algorithm).
% The exceptions are mostly the circuit optimizations, which naturally take a deeper look at the details of group operations. \han{We discuss..}\yun{At this point the reader could be even more confused: a possible previous reaction: `oh, so Regev's is not generic.'  At this point, a possible reaction: `what, is it?'}

We formally establish the quantum generic group model (QGGM) by restricting that access to group elements is provided through the group oracle. The QGGM resembles the classical generic group model (GGM) \ifnum\llncs=0~\cite{Shoup97,Mau05}\fi proposed for arguing the security of group-theoretic cryptographic problems in classical settings\ifnum\llncs=1~\cite{Shoup97,Mau05}\fi. 
As in the classical GGM, the main complexity measure in the QGGM is the number of group operation queries, which we refer to as the group operation complexity. \han{I try to avoid ``query complexity'' and use ``group operation complexity'' for clarity}
In addition, we are also concerned with the \emph{depth} of group operation queries and the quantum memory size to study the power of near-term quantum computers for the DL problem.

\paragraph{Lower bound in the fully quantum setting.} 
Our first result states that no generic quantum algorithm in the QGGM can solve the DL problem much faster than Shor's original algorithm.
% , even with parallel group operations. 
Precisely, we show the following theorem.
\begin{theorem}\label{thm: intro_DL}
    For a prime-order cyclic group $\cG$,
    any generic quantum algorithm solving the discrete logarithm problem over $\cG$ must make $\Omega(\log {|\cG|})$ group operation queries.
\end{theorem}

To establish this theorem, for any generic quantum DL algorithm $A$, we construct a generic \emph{classical} DL algorithm $B$ in the GGM that perfectly simulates the output of $A$.
Although the classical simulation may require unbounded time for precise simulation, its group operation complexity is only exponentially larger than that of $A$. 
Then, we observe that the known DL lower bound in the classical GGM of $\Omega(|\cG|^{1/2})$ group operations~\cite{Shoup97,Mau05} holds even if the algorithm is allowed to run in unbounded time.
We obtain the desired result by combining this fact with the above simulation with an exponential blowup. 
We note that the na{\"i}ve version of Shor's algorithm has the group operation complexity of $4\log |\cG|$ matching to the lower bound in~\Cref{thm: intro_DL} up to the constant multiplicative factor.
The same strategy establishes similar hardness of other group-theoretic problems, such as computational/decisional Diffie-Hellman problems.

% It is known that a generic DL algorithm in the classical GGM must make $\Omega(|\cG|^{1/2})$ classical group operation queries even if the algorithm is allowed to run in unbounded time~\cite{Shoup97,Mau05}.   
% Combined with the above simulation with an exponential blowup, we obtain the desired result. 
% We also establish the similar hardness of other group-theoretic problems, such as CDH and DDH using this simulation.
% We note that the na{\"i}ve version of Shor's algorithm\footnote{We describe Shor's algorithm in the QGGM in~\Cref{sec: algo} for completeness, with the other algorithms in the QGGM.} has the matching group operation complexity to the lower bound in~\Cref{thm: intro_DL}.

We also show that applying the group operation in parallel does not help much through the same proof.
\begin{theorem}\label{thm: intro_DL_parallel}
    For a prime-order cyclic group $\cG$,
    any generic quantum algorithm solving the discrete logarithm problem over $\cG$ must make group operation queries of depth $\Omega(\log {|\cG|})$.
\end{theorem}
The above result may initially seem sufficient to refute our main question. However, this is not the case because this lower bound only considers purely quantum algorithms, which even exclude classical preprocessing of group elements.

\paragraph{Hybrid quantum-classical algorithms.} 
We observe that some simple (combination of) folklore hybrid quantum-classical algorithms can do better than the purely quantum bound, exploiting classical computation to perform most group operations.

These hybrid algorithms consist of two phases: 
They first compute multiple group elements using 
$O(\mathrm{polylog}~|\cG|)$ classical group operation queries and store them as precomputed data. Then, 
they implement Shor's algorithm using the stored group elements using $O(\log|\cG|/\log\log |\cG|)$ quantum group operations and $O(\log\log |\cG|)$ quantum group operation depth (A more fine-grained tradeoff can be found in~\Cref{thm: base-p_Shor,thm: depth_Shor}).

We complement these algorithms by proving the matching lower bounds. 
We formalize a model for generic hybrid quantum-classical algorithms that captures the above algorithms and a more general class of algorithms.
In the model, we allow an algorithm to make both classical and quantum group operation queries with the restriction that it is \emph{forced to measure} all the registers whenever its quantum group operation number or depth count exceeds a certain threshold. 
It is supposed to capture hybrids of classical and quantum computers with limited coherence time. Note that we do not consider noises in our model, whereas actual near-term quantum computers are likely to be noisy. Since our main results are the lower bounds, this just makes our results stronger. 

The following theorem states the limitations of the generic hybrid algorithms, showing that the above hybrid algorithms are indeed optimal with respect to both group operation number and depth.

\begin{theorem}\label{thm:intro_hybrid}
    For a prime-order cyclic group $\cG,$ any generic hybrid quantum-classical algorithm solving the discrete logarithm problem with $O(\poly\log|\cG|)$ total group operation queries (including both classical and quantum) must make $\Omega(\log |\cG|/\log\log|\cG|)$ quantum group operation queries of depth $\Omega(\log\log |\cG|)$ 
    between some two consecutive forced measurements.

    More generally, any generic hybrid DL algorithm with $Q$ total group operations must make $\Omega(\log |\cG|/\log Q)$ quantum group operations of depth $\Omega(\log\log |\cG|-\log\log Q)$ between some two consecutive forced measurements.
\end{theorem}

\paragraph{Quantum memory-bounded algorithms.} 
Quantumly processable memory is an expensive resource, either quantum random-accessible memory that can store quantum states (QRAQM) or classical memory that stores classical data but can be accessed coherently (QRACM).\footnote{These two types of quantum-accessible memory were studied in~\cite{Kup13}. 
Formally, 
QRACM stores classical data $(x_i)_i$ and enables one to realize a unitary operation $\ket{i}\otimes \ket{0}\mapsto \ket{i}\otimes \ket{x_i}$.}
While the original Shor's algorithm only uses quantum memory that stores a single group element, the hybrid algorithms described above make use of relatively large quantum memory (\Cref{thm: depth_Shor}) or large QRACM (\Cref{thm: base-p_Shor}). This motivates the question of whether quantumly processable memory is necessary even for a mild speed-up of Shor's algorithm.

We prove that it is indeed necessary. We define a model for generic hybrid algorithms with memory constraints.
The following theorem asserts such a lower bound in the memory-bounded model.
\begin{theorem}\label{thm:intro_memory}
    For a prime-order cyclic $\cG,$ any generic hybrid algorithm solving the DL problem with 
    quantum memory that can store $t$ group elements   
    and
    no QRACM
    must make either $\Omega(\sqrt{|\cG|})$ classical or quantum group operation queries in total or $\Omega(\log |\cG|/\log t)$ quantum group operations between some two consecutive forced measurements.\footnote{This gives a depth lower bound of $\Omega(\log |\cG|/t \log t)$ as an immediate corollary as an algorithm can make at most $t$ queries in one parallel group operations in this setting.}

    More generally, any generic hybrid DL algorithm with 
    quantum memory that can store $t$ group elements   
    and
    QRACM that can store $r$ group elements
    must make either $\Omega(\sqrt{|\cG|})$ group operations in total or $\Omega(\log |\cG|/\log (tr))$ quantum group operation queries between some two consecutive forced measurements. 
\end{theorem}
In particular, the above theorem implies that classical queries cannot reduce the number of quantum queries beyond $\Omega(\log |\cG|/\log t)$, or just $\Omega(\log |\cG|)$ when $t=O(1).$
We have algorithms that match the above lower bounds: Baby-step giant-step algorithm
makes $O(\sqrt{|\cG|})$ classical group operations, and the hybrid algorithm in \Cref{thm: base-p_Shor} with quantum memory that can store $t$ group elements and no QRACM makes $\Omega(\log |\cG|/\log t)$ quantum queries.

\paragraph{The multiple DL problem.} 
The multiple discrete logarithm problem asks to solve multiple instances of the DL problem with the same underlying group simultaneously. When $m$ DL instances are given, this problem is called $m$-MDL. This problem is particularly interesting in the context of the standard curves in elliptic curve cryptography, where only a few curves are recommended as standard.
Classically, Kuhn and Struik~\cite{KS01} suggested an $O(\sqrt{m |\cG|})$ generic algorithm for the $m$-MDL problem, and Yun~\cite{Yun15} proved the matching lower bound.

In~\Cref{thm: MDL_alg}, we present a generic quantum algorithm for the multiple discrete logarithm problem using the results in vectorial addition chain~\cite{Pip80}. 
If $\log m /\log |\cG|=o(1)$ and $m=\Omega(\log |\cG|)$, it solves the $m$-MDL problem using $O(m\log |\cG|/\log(m) )$ group operations. This gives an amortized group operation complexity of $O(\log |\cG|/\log m)$ per DL instance.

Regarding~\Cref{thm: intro_DL}, the complexity of the $m$-MDL problem is lower than solving each instance individually. It is related to the quantum annoying property~\cite{Tho19,ES21} suggested in the context of password-authenticated key exchange (PAKE), which roughly means that quantum algorithms must solve a DLP for each password guess of PAKE. Our algorithm shows that the strongest form of quantum annoying cannot hold, regardless of the PAKE construction.

We can derive the lower bound of the $m$-MDL problem similarly to~\Cref{thm: intro_DL} and using the classical lower bound given in~\cite{Yun15}. However, this would only give a lower bound of  $\Omega(\log m + \log |\cG|)$ group operations.
So there is an apparent gap between the upper and lower bounds from our approach. 
% We leave a more accurate generic asymptotic complexity of $m$-MDL as an open problem.

\subsection{Discussion}\label{sec:discussion}
\han{Because of Regev-style algorithms, this paragraph is more important than before. I moved it here and modified it with a focus on Regev, and expanded the impracticality of HS05. Other parts are almost unchanged.}
\paragraph{QGGM vs. non-generic improvements.}
We discuss two non-generic complexity improvements of the DL algorithms. 
\begin{itemize}
    \item Following the recent better quantum factoring algorithm~\cite{Regev23,RV23}, Eker{\aa} and G{\"a}rtner~\cite{EG23} gave a better quantum DL algorithm. 
    Given the QGGM lower bounds in this paper, one may wonder if their algorithms are non-generic because otherwise, it seems like a contradiction.
    We remark that they \emph{can be described} in the QGGM, and they indeed obey our group operation complexity lower bounds if we count the number of group operations.
    
    Their primary complexity measure is \emph{circuit complexity}\footnote{More precisely, the asymptotic circuit complexity. The practical implications of these algorithms are under debate, and we refer~\cite{EG24} for a recent discussion.}, not group operation complexity. For this purpose, they use the fact that group operations between two \emph{small} elements are much faster than ordinary group operations, while we assume that they both incur the same cost.\yun{'while we assume that they both incur the same cost'?  Incur is transitive.} On the other hand, using small group elements is only for faster group operations, and the correctness seems irrelevant to the smallness of base group elements.\han{I roughly checked this, and it's probably right. but I'm not 100\% confident about this. ``seems irrelevant'' is okay?}
    \item H{\o}yer and Spalek \cite{HS05} showed that the DL problem on $\mathbb Z_N$ can be solved by a hybrid quantum-classical algorithm with a constant quantum depth if we allow for unbounded fan-out gates.\footnote{It does not contradict the depth lower bound of the quantum Fourier transform~\cite{CW00}, which assumes that each gate acts on a constant number of qubits.} 
    This overcomes our quantum depth lower bound in \Cref{thm:intro_hybrid}.\footnote{Using fan-out gates does not affect the query depth in the QGGM.} 
    
    This is possible because their algorithm is non-generic. For example, they use that multiplication of many elements of $\mathbb Z_N$ can be done in $\mathbf{TC}_0$, i.e., computed by a constant depth classical circuit with threshold gates~\cite{SBKH93}. 
    We also note that they mainly focus on theoretic depth optimization and are unlikely to be practical for two reasons. The unbounded fan-out gates are believed to be hard to implement, so it's barely considered in near-term quantum devices. Further, even equipped with unbounded fan-out, the circuit size is increased to reduce the depth, making the algorithm require huge quantum memory.
\end{itemize}

Besides the above algorithms, to our knowledge, all non-generic quantum algorithms for the DL problem are circuit optimization of (variants of) Shor's (generic) DL algorithm~\cite{PZ03,RS14,RNSL17,HJNRS20}.
These optimizations leverage specific encoding structures for practical purposes, and the asymptotic complexity remains unchanged.

This circumstance is reminiscent of the classical GGM, where some non-generic algorithms, such as index calculus, show better efficiency than generic algorithms by exploiting the integer encoding of group elements. 
Still, the classical GGM has been used as a meaningful model for arguing the hardness of group-theoretic problems, especially for the general elliptic curves.
Thus, we believe that lower bounds in the QGGM are at least as meaningful as those in classical GGM. 

\paragraph{Practical implication.} Optimizing and estimating quantum attacks have been studied extensively to make them practical. Shor's algorithm is of practical interest and may be used for estimating\yun{estimating?} the deadline of mandating migration to post-quantum cryptography.
In particular, a recent estimation by Gidney and Eker{\aa}~\cite{GE21} predicts that a single real-world DL instance can be solved within a half day using millions of noisy qubits by a highly optimized quantum algorithm, assuming several plausible physical assumptions. 
Our result indicates that there is a fundamental limitation for the generic approaches.
Putting it differently, there are only a few ways to improve the quantum DL algorithms: reducing the hidden factors in the QGGM as in~\cite{Regev23} or optimizing the quantum computer itself or the circuits, unless a non-generic quantum algorithm is discovered.
% Our result indicates that optimizing the quantum computer itself or the circuits is the only way to improve quantum attacks on the discrete logarithm, unless a non-generic quantum algorithm is discovered.

\paragraph{Tight group operation complexity.}
Our lower bounds show asymptotically tight group operation complexity, but the constant factor has room for improvement. In the formal theorems, the concrete quantum query bounds are $0.25\log |\cG|+O(1)$ (or depth) in the fully quantum case (\Cref{thm: QGGM_DL}) and $\frac 13 \log |\cG|+O(1)$ for the memory bounded hybrid case with $t=r=1$ (\Cref{thm: memory DL and more}).\footnote{Regarding the constant factor, we notice that the specifications for the model of group operations matter a lot. For example, if we only allow $g,h\mapsto g\cdot h$ and not $g\cdot h^{-1}$, the memory-bounded case with $t=r=1$ becomes $0.5 \log |\cG| + O(1)$.}\han{I modified the constant factor and added a footnote.}
Shor's DL algorithm and early variants~\cite{ME98,Kitaev96} make \if\llncs=0 quantum and classical\fi group operations $2\log |\cG|$ times each, having a gap in the constant factor.

The hybrid quantum-classical algorithms~\cite{Kaliski17,Ekera19,Ekera21,EH17} narrow down this gap.
These algorithms solve the DL problem by repeating a certain procedure with $\log|\cG|+O(1)$ group operations about $\log^{O(1)}|\cG|$ times,\footnote{Precisely, Kaliski's algorithm~\cite{Kaliski17} repeats a subroutine of $\log |\cG|+1$ group operations $O(\log^{1.5}|\cG|)$ times, and Eker{\aa}'s algorithm repeat subroutines with $(1+1/s)\log |\cG|$ group operations about $s$ times for some bounded $s$~\cite{Ekera21}.} with appropriate classical pre- and post-processing. The constant gap still exists besides the number of subroutine calls. Filling this gap is an interesting open problem.
% \footnote{Regev reported a new efficient hybrid factoring algorithm, and remarked that the same idea may work for the DL problem~\cite{Regev23}. For clarity, we remark that this (hypothetical) DL algorithm is generic; it does not contradict our results---the complexity measures are different. The speed-up comes from the fact that multiplying small numbers is faster than multiplying large ones.}\han{I added this footnote}

Another interesting tradeoff point in our lower bound is the hybrid case (without memory bound) in~\Cref{thm: hybrid DL and more}. We may ask if a small number of quantum group operations could reduce the classical group operation queries. This theorem says that if a generic hybrid algorithm makes a single quantum group operation, then it should make $\Omega(|\cG|^{0.25})$ classical group operations. In other words, this does not rule out a hybrid DL algorithm with $|\cG|^{0.25}$ classical group operations and a single quantum group operation, which we do not know how to do. \Cref{thm: memory DL and more} rules out this case if there is a memory constraint.

The quantum complexity of the composite-order DL and MDL problems is also unknown. We do not know how to use the composite order either in constructing algorithms or proving lower bounds. We note that a recent work~\cite{Hhan24} resolves these problems by showing tight lower bounds, though in a slightly weaker hybrid model than ours.

\paragraph{Maurer-style vs. Shoup-style QGGM.}
In the classical setting, there are two formalizations of the GGM, one by Shoup~\cite{Shoup97} and the other by Maurer~\cite{Mau05}.
In Shoup's GGM, generic algorithms are given random labels of group elements and can perform group operations by sending labels to the oracle. On the other hand, in Maurer's GGM, all group elements are kept by the oracle, and generic algorithms can access them only through group operation or equality check queries.  These two GGMs are known to be equivalent for "single-stage games," which include the DL and related problems considered in this paper~\cite{Zhandry22a}.\footnote{Jager and Schwenk~\cite{JS08} originally claimed a general equivalence, but Maurer, Portmann, and Zhu~\cite{MPZ20} pointed out a counterexample. Zhandry~\cite{Zhandry22a} resolved this issue by reproving the equivalence in the case of single-stage games. Precisely speaking, he proved equivalence between Shoup's GGM and what is called the type-safe model, which is a variant of Maurer's GGM for single-stage games, but there is no difference between the type-safe model and Maurer's GGM when we consider group-theoretic problems such as the DL problem.   
}

Our QGGM is defined as a quantum analog of Maurer's GGM. It is possible to define it in Shoup's style. Indeed, such a model was already considered in \cite{Zhandry21} under the name of "post-quantum GGM." 
It is easy to show that any generic algorithm that works in our (Maurer-style) QGGM also works in Shoup-style QGGM. 
On the other hand, it seems difficult to show the other direction in the quantum setting, even if we focus on single-stage games. 
Thus, it would make our results stronger if we could prove similar lower bounds in Shoup-style QGGM. 
We believe that the lower bound in the fully quantum setting
(\Cref{thm: intro_DL}) can be extended to Shoup-style QGGM with a similar proof if the label space is much larger than the group order. 
On the other hand, we do not know how to generalize the lower bounds for hybrid algorithms (\Cref{thm:intro_hybrid,thm:intro_memory}) to Shoup-style QGGM.
For this reason, we focus on Maurer-style QGGM in this paper. 
We believe that lower bounds in Maurer-style QGGM are still meaningful, given that it captures Shor's algorithm and many variants.

\paragraph{Hidden subgroup problems and other potential directions.} 
This paper suggests the number of (quantum) group operations as a complexity measure for studying the DL and related problems. We discuss the potential applications to the hidden subgroup problem (HSP).

In the hidden subgroup problem (HSP) literature, the primary complexity measure is the query complexity to the oracle function $h:G\to X$ hiding a subgroup, i.e. $h(g_1)=h(g_2)$ iff $g_1 H= g_2 H$ for some hidden subgroup $H$ of $G$. The standard approach, or \emph{Fourier sampling}, to the HSP over abelian groups makes a single oracle query to $h$. This approach is also a subroutine to solve some HSPs over non-abelian groups, e.g., in \cite{HRT00,EH00,Kuperberg05}. Finally, it is shown that $O(\log^4 |\cG|)$ queries to $h$ suffice for the HSP over an arbitrary group~\cite{EHK04}. This makes proving lower bounds in terms of query complexity unlikely to yield superpolynomial lower bounds for the HSP.\footnote{For a certain restricted class of algorithms, there are some known limitations~\cite{MRS08,HMRRS10}.}

Interestingly, these HSP algorithms can be considered generic algorithms by extending our QGGM for general groups.
% \takashi{Stritly speaking, we haven't defined QGGM for groups other than cyclic groups.} 
Also, contrary to the query complexity (to $h$), the group operation complexity of~\cite{EHK04} is exponentially large. One may wonder if the group operation complexity can provide an interesting lower bound of the HSP for some nonabelian groups.
The full answer is elusive with this paper's tools.
The \emph{dihedral group} case, a crucial case regarding its connection to the lattice-based~\cite{Regev04} and isogeny-based cryptography~\cite{Peikert20,CJS14}, has a negative answer to this question, as the algorithm of Ettinger and H{\o}yer~\cite{EH00} only makes a polynomial number of group operations.

We believe exploring other potential applications of the generic model presented in this paper is an interesting topic. For example, can we argue something about factoring by extending our model to the ring operations?\footnote{In~\cite{Hhan24}, some progress is made in this direction.}
% Are there any cryptographic or complexity-theoretic applications of the generic algorithms?\footnote{We note that the oracle separation of $\sf QMA$ and other classes is studied with the group oracles. For example, Watrous~\cite{Watrous00} separated $\sf QMA$ from $\sf MA$ using the group non-membership (GNM) problem, building on~\cite{Babai92}. Later, Aaronson and Kuperberg~\cite{AK07} argue that it is pointless to separate $\sf QMA$ and $\sf QCMA$ regarding the oracle query complexity, using the polynomial-query HSP algorithm~\cite{EHK04} as a subroutine. However, their algorithm executes the HSP algorithm on the model group, not the given input group. Thus, their algorithm in principle has a polynomial group operation complexity regarding the input group.\han{I think this footnote is very vague and hard to understand. We may just remove this comment (or the paragraph itself).}\takashi{I couldn't understand this. We may remove this footnote and the final sentence of the paragraph.}}

\subsection{Related Works}
\paragraph{Post-quantum GGM.}   
Zhandry~\cite{Zhandry21} introduced a model called post-quantum GGM as a quantum analog of Shoup's GGM. He showed that the generic group oracle in the model is quantumly reset indifferentiable from ideal ciphers. This means that generic groups can be used to construct symmetric key encryption secure against quantum adversaries.  
On the other hand, the work does not discuss the hardness of the DL and its related problems in the post-quantum GGM.  
\yun{A generic group and an ideal cipher cannot be indifferentiable.  Perhaps... I'm not sure...  'indifferentiably equivalent'?  In the sense that using A you can construct something indifferentiable to B, and using B you can construct something indifferentiable to A?}
\han{I think in a rigorous sense you are right, especially Zhandry only showed that (in my understanding) the generic group is reset indifferentiable from the random injective functions (not ideal cipher!), but I want to follow the statement in Zhandry's abstract. I think nobody will complain (or carefully read) about this one.}

\paragraph{Generic group action model.}
While the DL problem on cyclic groups can be solved in quantum polynomial time by Shor's algorithm, the DL problem for \emph{group actions} is believed to be hard against quantum computers. Such group actions with the quantum hardness of the DL problem have been used as bases of some proposals of post-quantum cryptography~\cite{Couveignes06,RS06,CLMPR18,JQSY19}. 
Montgomery and Zhandry~\cite{MZ23} and Duman et al.~\cite{DHKKLR23} introduced generic models for group actions and studied the relations between the DL and related problems. We stress that their results are not proving the lower bounds.\footnote{Indeed, \cite{DHKKLR23} argued that we could not hope for the superpolynomial lower bound of the DL problem in group actions due to~\cite{EH00}, similar to our discussion on the dihedral HSP.}

\paragraph{Hybrid quantum-classical algorithms.}
The hybrid quantum-classical algorithms have recently begun to attract more attention in various aspects. In~\cite{CCL23,CM20}, the authors studied the relations between the hybrid algorithm with shallow quantum circuits and $\sf BQP$, refuting the conjecture of Josza~\cite{Jozsa06} and proving Aaronson's conjecture~\cite{Aar05}. 
The study of hybrid algorithms with shallow quantum circuits was continued in~\cite{ACCGSW23} relative to random oracles.
\cite{Rosmanis22} studied the hybrid algorithm in the context of Grover's algorithm, showing that classical queries cannot assist quantum computation. \cite{HLS22} further developed the tools for hybrid algorithms with random oracles and showed a similar result for collision finding. 
Our model of generic hybrid algorithms is inspired by~\cite{CCL23,ACCGSW23}, as well as the other papers.

\section{Technical Overview}\label{sec: overview}
\paragraph{Classical GGM.}
First, we recall the classical GGM as formalized by Maurer~\cite{Mau05}. Let $\cG$ be a cyclic group of order $N$ with a generator $g$ in which we consider group-theoretic problems such as the DL problem. A generic algorithm $A$ is formalized as an oracle-aided algorithm that has classical access to an oracle, which keeps a table $T$ storing elements of $\mathbb Z_N$. At the beginning, when $A$ takes $g^{y_1},...,g^{y_m}$ as input, the table $T$ is initialized as $(y_1,...,y_m,0,...,0)$.\footnote{The size of $T$ can be unbounded.}
% \han{Do we need to write $g^{y_i}$, etc here?}
%, but we treat it as having a fixed size for convenience.}  
The generic algorithm $A$ can make the following two types of queries:
\begin{itemize}
    \item \emph{Group operation queries.} 
    When $A$ submits  
    $(b,i,j,k)\in \bit \times \N^3$, the oracle finds $i$-th element $x_i$ and $j$-th element $x_j$ in the table $T$ and overwrites the $k$-th element of $T$ by $x_i + (-1)^b x_j$. 
    Nothing is returned to $A$.   
    \item \emph{Equality queries.} 
    When $A$ submits  
    $(i,j)\in \N^2$, 
    the oracle returns $1$ if $i$-th and $j$-th elements of $T$ are equal and otherwise returns $0$.  
\end{itemize}
We only count the number of group operation queries and allow equality queries for free, following the previous models~\cite{Mau05,Zhandry22a}.\footnote{See~\Cref{sec: Model} for more discussion.}
%\takashi{Can you add some explanation?} 
% \yun{The footnote is a very good explanation.  How about sending this right after the main definition?}

Finally, $A$ outputs a bit string or an index $i^*$ of $T$.
In the latter case, $g^{x_{i^*}}$ is treated as $A$'s output where $x_{i^*}$ is the $i^*$-th element in $T$.  

\paragraph{Quantum GGM.}
We define the Quantum GGM (QGGM) as a natural quantum analog of the classical GGM where 
$A$ is allowed to make quantum queries and the table $T$ is stored in a quantum register $\mathbf{T}$.  
However, since overwriting values of quantum registers is not unitary, we formalize group operation queries in a slightly different way. Specifically, a group operation query is (a superposition of) $(b,i,j)\in \bit \times \N^2$ and the oracle replaces $i$-th element of the table register $\mathbf{T}$ with $x_i + (-1)^b x_j$ (in superposition) where $x_i$ and $x_j$ are $i$-th and $j$-th elements of $\mathbf{T}$ before the query, respectively.  
In this way, we can ensure that it is a unitary operation. For clarity, we describe how the oracle works for group operation and equality queries where $\mathbf{Q}$ is the query register: 
\begin{itemize}
    \item \emph{Group operation queries.} 
   Apply the following unitary on  $\mathbf{Q}$ and  $\mathbf{T}$:
    \begin{equation*}
    \ket{b,i,j}_{\mathbf{Q}} \otimes \ket{...,x_i,...,x_j,...}_\mathbf{T} \mapsto
    \ket{b,i,j}_{\mathbf{Q}} \otimes \ket{...,x_i + (-1)^b x_j,...,x_j,...}_{\mathbf{T}}
    \end{equation*}
    if $i\neq j$ and otherwise it does nothing.
    \item \emph{Equality queries.} Apply the following operation on  $\mathbf{Q}$ and  $\mathbf{T}$:
    \begin{equation*}
        \ket{b,i,j}_{\mathbf{Q}} \otimes \ket{...,x_i,...,x_j,...}_{\mathbf{T}} \mapsto\ket{b\oplus t,i,j}_{\mathbf{Q}}\ket{...,x_i,...,x_j,...}_{\mathbf{T}}
    \end{equation*}
    where $t=1$ if $x_i=x_j$ and $t=0$ otherwise.
\end{itemize}
Initialization and finalization of a generic algorithm are exactly the same as in the classical GGM except that $\mathbf{T}$ is measured in the computational basis at the end.

\paragraph{Basic idea: the fully quantum setting.}
Our idea is to simulate a generic algorithm $A$ taking $m$ group elements as input in the QGGM by a generic algorithm $B$ in the \emph{classical GGM} with an exponential blowup in the number of group operations (or simply, queries). Since we have a group operation complexity lower bound of $\Omega(|\cG|^{1/2})$ for the DL problem in the classical GGM, such a simulation gives a lower bound of $\Omega(\log |\cG|)$ in the QGGM. In particular, the classical lower bound holds even for unbounded algorithms as long as the condition on the number of queries is satisfied.

The idea for the simulation is extremely simple. At the beginning, the table register $\mathbf{T}$ of the QGGM has $m$ non-zero elements $(y_1,...,y_m)$. %where $(g^{y_1},...,g^{y_m})$ is the input of $A$. 
Suppose that $A$ makes one (potentially parallel) quantum group operation query. After the query, $\mathbf{T}$ can only contain elements of the form 
$z_1y_1+z_2y_2+...+z_my_m$
where $|z_i|\le 1$ for all $i\in [m]$ 
%$y_i$ or $y_i+(-1)^b y_j$ for some $i,j\in [m]$ and $b\in \bit$  
in any branch with a non-zero amplitude. 
%In particular, all of them can be written of the form $z_1y_1+z_2y_2+...+z_my_m$ where $|z_i|\le 1$ for all $i\in [m]$. 
After $A$ makes the next (potentially parallel) quantum group operation query,  a similar argument shows that $\mathbf{T}$ can only contain elements of the form $z_1y_1+z_2y_2+...+z_my_m$ where $|z_i|\le 2$ for all $i\in [m]$. 
By repeating a similar argument recursively, one can see that after $d$-layer of parallel quantum group operation queries, $\mathbf{T}$ can only contain elements of the form 
$z_1y_1+z_2y_2+...+z_my_m$ where $|z_i|\le 2^{d-1}$ for all $i\in [m]$. In particular, the number of such elements is at most  $(2^{d}+1)^m\le 2^{m(d+1)}$. 
Thus, if the generic algorithm $B$ in the classical GGM creates all these elements in its table in advance using $2^{m(d+1)}$ classical queries, it can perfectly simulate $\mathbf{T}$ for $A$. 
Note that $B$ can run in unbounded time, though it only makes classical queries. In particular, it can simulate any quantum superposition of the group elements in the table.
% Note that $B$ can run in unbounded time and, in particular, can simulate any quantum computation of these group elements by brute force, though its queries are limited to be classical. 
% For creating the above elements in the table, $B$ only needs to make at most $2^{m(d+1)}$ classical queries. 
This means that a generic algorithm of query depth $d$ in the QGGM can be perfectly simulated by a generic algorithm that makes $2^{m(d+1)}$ queries in the classical GGM. In particular, for the DL problem, we have $m=2$ since the input is $y_1=g$ and $y_2=g^x$ for random $x$.  Combined with the lower bound in the classical GGM, we obtain \Cref{thm: intro_DL}. 

\paragraph{Hybrid quantum-classical algorithms.}
\han{Should we change the order of depth and query here?}
We explain how to extend the above idea to the hybrid quantum-classical algorithms. First, we describe our formalization of hybrid quantum-classical algorithms in the QGGM. 
A hybrid quantum-classical algorithm is characterized as a consecutive execution of quantum \emph{subroutines} $U_1,...,U_T$ followed by a classical post-processing algorithm $A_{T+1}$. 
Each subroutine $U_i$ makes arbitrarily many classical queries and a bounded number or depth of quantum queries and measures all the registers including the table register $\mathbf{T}$ at the end. 
% \takashi{Strictly speaking, I guess the final measurement is not a part of $U_i$. But is this explanation okay?}\han{I included it as a part of $U_i$. Is there any problem with this?}
%After applying each subroutine $U_i$, all the registers including the table register $\mathbf{T}$ must be measured.
$A_{T+1}$ makes arbitrarily many classical queries and no quantum query. 

We first consider the depth-bounded case where each subroutine can have quantum query depth at most $d$. Let $Q$ be the total number of the hybrid algorithm's queries including both classical and quantum ones. The idea is similar to the basic case: The classical simulation algorithm in GGM creates all the group elements that may appear in the table register $\mathbf T$.
% We count the number of all group elements that may appear in the table register $\mathbf{T}$ and the simulation algorithm in the classical GGM creates all of them in its classical table. 

We first analyze each subroutine and then apply an inductive argument. 
Suppose that a subroutine $U_i$ is described as a sequence 
\ifnum\llncs=1
\[(C_{i,0},O_{i,1},...,C_{i,d-1},O_{i,d})\]
\else
$(C_{i,0},O_{i,1},...,C_{i,d-1},O_{i,d})$ 
\fi
where each $C_{i,j}$ only makes classical queries and each $O_{i,j}$ makes one parallel quantum query. 
Let $Q_i$ be the total number of queries made by $U_i$ and
let $c_{i,j}$ be the number of classical queries made by $C_{i,j}$.  
Let $m_i$ be the number of non-zero elements stored in the table register $\mathbf{T}$ when $U_i$ starts. 
For $j=0,1,...,d-1$, let $S_{i,j}\subseteq \mathbb Z_N$ be the set of elements that appear in the table register $\mathbf{T}$ in some branch with a non-zero amplitude right before the application of $O_{i,j+1}$ and let $S_{i,d} \subseteq \mathbb Z_N$ be the set right after the application of $O_{i,d}$ (before the forced measurement). 
%At the beginning of any subroutine, the table register $\mathbf{T}$ collapses to a classical state and contains at most $m+Q+1$ elements (including $0$) since $\mathbf{T}$ has $m+1$ elements at the beginning and each classical or quantum query just adds at most one new element.  
In our model, we can show that $|S_{i,0}|\le m_i+1+c_{i,0}$,  
    $|S_{i,j}| \le 2|S_{i,j-1}|^2 + c_{i,j}$, 
    and $|S_{i,d}| \le 2|S_{i,d-1}|^2$.
    Thus, we have 
    \[
       |S_{i,d}|\le  2^{2^d} \left( m_i +1+ c_{i,0} + c_{i,1} + ...+ c_{i,{d-1}} \right)^{2^d} \le 2^{2^d} \left( m_i + Q_i+1\right)^{2^d}
    \]
Then, by a similar argument to the basic case, a generic algorithm in the classical GGM can simulate the subroutine $U_i$ by making at most $2^{2^d}\left(m_i + Q_i+1\right)^{2^d}$ group operation queries. 
Moreover, it is easy to see that we have $m_i\le m+ Q_1+...+Q_{i-1}$ where $m$ is the number of elements given as input. This is because one classical or quantum group operation only adds at most one new element in the table register. 
Then, if we let $c$ be the number of classical queries by $A_{T+1}$, 
the total number of classical queries needed to simulate the whole execution is at most 
\[
    2^{2^d}(m+Q_1+1)^{2^d}+ 2^{2^d}(m+Q_1+Q_2+1)^{2^d} + ... + c \le Q + T \cdot 2^{2^d}(m+Q+1)^{2^d}
\]
where we use $Q_1+Q_2+...+Q_T+c\le Q$. 
That is, any generic hybrid algorithm with the total number of queries $Q$ and bounded quantum query depth $d$ can be simulated by a generic algorithm in the classical GGM that makes at most  $Q + T \cdot 2^{2^d}(m+Q+1)^{2^d}$ classical queries. Combined with the classical GGM lower bound, we obtain the depth-bounded part of \Cref{thm:intro_hybrid}.

Next, we consider the query-bounded case where 
each subroutine can make at most $q$ quantum queries, and the total number of the hybrid algorithm's queries is $Q$, including both classical and quantum ones.
The idea is similar to the depth-bounded case, but we have to count the number of elements that appear in the table register $\mathbf{T}$ more carefully by making use of the fact that there are no parallel queries. 
We use similar notations to the depth-bounded case where the difference is that each $O_{i,j}$ makes only one non-parallel group operation query instead of a parallel one and the index $j$ ranges in $[q]$ instead of $[d]$. First, we remark that we can simulate $C_{i,0},...,C_{i,q-1}$ by making at most $Q_i$ classical queries in an obvious way. Thus, we ignore them in the following analysis and simply add $Q_i$ to the number of classical queries needed to simulate the subroutine at the end.  Recall that $m_i$ denotes the number of non-zero elements stored in $\mathbf{T}$ at the beginning of the subroutine $U_i$.
%First, the subroutine $U_i$ applies $C_{i,0}$, after which $\mathbf{T}$ contains $m_i+c_{i,0}+1$ elements including $0$. We can simulate $C_{i,0}$ by making $c_{i,0}$ classical queries in an obvious way. 
Before the subroutine $U_i$ applies $O_{i,1}$, $\mathbf{T}$ has a classical state that has at most $m_i+Q_i+1$ elements including $0$. 
After applying $O_{i,1}$,  $\mathbf{T}$ is a superposition of at most $2(m_i+Q_i+1)^2$ different tables. We observe that in each possible table with a non-zero amplitude, at most one new element is added. Thus, we can simulate $O_{i,1}$ by making at most $2(m_i+Q_i+1)^2$ classical queries. Next, for each fixed branch of $\mathbf{T}$, we can do the same analysis to see that we can simulate $O_{i,2}$ by making at most $2(m_i+Q_i+1)^2$ classical queries. Since there are $2(m_i+Q_i+1)^2$  branches, the total number of classical queries needed to simulate $O_{i,2}$ is at most $2(m_i+Q_i+1)^2\cdot 2(m_i+Q_i+1)^2=4(m_i+Q_i+1)^4$ and 
we have at most $4(m_i+Q_i+1)^4$ branches with a non-zero amplitude. 
By repeating a similar argument recursively, we can simulate $O_{i,j}$ by making $2^j(m_i+Q_i+1)^{2j}$ classical queries. Thus, the total number of classical queries needed to simulate the subroutine $U_i=(C_{i,0},U_{i,1},...,C_{i,q-1},U_{i,q})$ is at most 
\[
Q_i + 2(m_i+Q_i+1)^2 +... + 2^q (m_i+Q_i+1)^{2q} \le  Q_i+2^{q+1}(m_i+Q_i+1)^{2q}. 
\]

\ifnum\llncs=0
Noting that we have $m_i\le m+ Q_1+...+Q_{i-1}$, 
the total number of classical queries to simulate the whole hybrid algorithm is at most
 \[
    \left(Q_1+2^{q+1}(m+Q_1+1)^{2q}\right)+\left(Q_2+2^{q+1}(m+Q_1+Q_2+1)^{2q}\right)+ ...+c\le Q+T\cdot 2^{q+1} (m+Q+1)^{2q},
\]
\else
Noting that $m_i+Q_i\le (m+ Q_1+...+Q_{i-1})+Q_i \le m+Q$, 
the total number of classical queries to simulate the whole algorithm is at most
\[
\sum_{i=1}^T \left(Q_i+2^{q+1}(m+Q+1)^{2q}\right)
% \left(Q_1+2^{q+1}(m+Q_1+1)^{2q}\right)+\left(Q_2+2^{q+1}(m+Q_1+Q_2+1)^{2q}\right)+ ...+
+c\le Q+T\cdot 2^{q+1} (m+Q+1)^{2q},
\]
\fi
where we use $Q_1+Q_2+...+Q_T+c\le Q$. 
That is, any generic hybrid algorithm with the total number of queries $Q$ and bounded quantum query number $q$ can be simulated by a generic algorithm in the classical GGM that makes at most  $Q + T \cdot 2^{q+1}(m+Q+1)^{2q}$ classical queries. Combined with the lower bound in the classical GGM, we obtain the query-bounded part of \Cref{thm:intro_hybrid}. 

\paragraph{Quantum-memory-bounded algorithms.}
To formalize a memory-bounded generic quantum algorithm, we divide the table register $\mathbf{T}$ into the quantum part  $\mathbf{T}_Q$ and classical part $\mathbf{T}_C$. We restrict  $\mathbf{T}_Q$ to store at most $t$ elements whereas  $\mathbf{T}_C$ can store arbitrarily many elements. 
A generic hybrid quantum-classical algorithm without QRACM cannot send a group operation query that involves a superposition over indices in $\mathbf{T}_C$. In this setting, the number of new elements that may be computed by one quantum query is at most $2t(t-1)$. Thus, by a similar analysis to the query-bounded case for quantum-memory-bounded generic algorithms in the previous paragraph, we can see that the number of classical queries to simulate each subroutine $U_i$ is at most 
\[
Q_i + (2t(t-1)+1) +... + (2t(t-1)+1)^{q} \le  Q_i+2\cdot(2t(t-1)+1)^{q}, 
\]
where $+1$ appears for a technical reason,\footnote{Slightly precisely, the element $0$ in the memory-bounded simulation could be removed from the memory during the simulation.}
\han{I corrected the above and added a footnote.}
\ifnum\llncs=0
and thus the total number of classical queries to simulate the whole hybrid algorithm is at most
 \[
    \left(Q_1+2\cdot(2t(t-1)+1)^{q}\right)+\left(Q_2+2\cdot(2t(t-1)+1)^{q}\right)+ ...+c\le Q+2T (2t(t-1)+1)^{q}.
\]
Combined with the classical lower bounds,
this implies the former part of \Cref{thm:intro_memory}.
\else
and the total number of classical queries for the simulation is at most
 \[
    \sum_{i=1}^{T} \left(Q_i+2\cdot(2t(t-1)+1)^{q}\right)+c\le Q+2T\cdot (2t(t-1)+1)^{q}.
\]
We have the former part of \Cref{thm:intro_memory} by the classical lower bounds.
\fi

For capturing QRACM that can store $r$ elements, we allow a generic hybrid algorithm to make a query involving a superposition over indices in $\mathbf{T}_C$ as long as the number of indices in $\mathbf{T}_C$ involved in the superposition is at most $r$. In this setting, the number of new elements that may be computed by one quantum query is at most $2t\cdot(t-1+r)$. 
\ifnum\llncs=0
Thus, by a similar analysis where we replace $2t(t-1)$ with $2t\cdot (t-1+r)$, 
the total number of classical queries to simulate the whole hybrid algorithm is at most
 \[
    \left(Q_1+2\cdot(2t(t+r-1)+1)^{q}\right)+\left(Q_2+2\cdot(2t(t+r-1)+1)^{q}\right)+ ...+c\le Q+2T\cdot (2t(t+r-1)+1)^{q}.
\]
Combined with the classical lower bound,
this implies the latter part of \Cref{thm:intro_memory}.
\else
By a similar analysis where we replace $2t(t-1)$ with $2t\cdot (t-1+r)$, 
the total number of classical queries to simulate the whole hybrid algorithm is at most
 \[
    \sum_{i=1}^T \left(Q_i+2\cdot(2t(t+r-1)+1)^{q}\right)+c\le Q+2T\cdot (2t(t+r-1)+1)^{q}.
\]
We have the latter part of \Cref{thm:intro_memory} by the classical lower bounds.
\fi
\section{The Adversarial Model}\label{sec: Model}
This section defines the model of \emph{generic} adversaries for the discrete logarithm and related problems. 
\Cref{subsec: GGM} defines the generic group model for the classical and quantum adversaries and \Cref{subsec: Problems_in_QGGM} summarizes the cryptographic problems such as the discrete logarithm and their lower bounds in the classical generic group model.

\subsection{The Generic Group Models}\label{subsec: GGM}
\paragraph{Classical generic group model.} 
We first review the \emph{classical} generic group model (GGM) as defined in~\cite{Mau05}.
A generic algorithm $A$ in the GGM interacts with an oracle that keeps a function $T:\N\rightarrow \gZ_N$ for some positive integers $N$.   
We often regard $T$ as a table consisting of group elements, and we often refer to $T(i)$ by the $i$-th element in the table $T$. 
At the beginning, $T$ is initialized as 
$T(i):=y_{i}$ for $i\in[m]$ and $T(i):=0$ for all $i> m$ where 
$(y_1,...,y_{m})\in \gZ_N^m$ is the input of $A$.  
$A$ is allowed to make the following queries:
\begin{itemize}
    \item \emph{Group operation queries.} 
    When $A$ submits  
    $(b,i,j,k)\in \bit \times \N^3$, the oracle overwrites $T(k):=T(i) + (-1)^b T(j)$.    
    Nothing is returned to $A$.   
    \item \emph{Equality queries.} 
    When $A$ submits  
    $(i,j)\in \N^2$, 
    the oracle returns $1$ if $T(i)=T(j)$, and $0$ otherwise. 
\end{itemize}
Finally, $A$ outputs a classical string or a special symbol $\mathsf{group}$ along with an integer $i$. In the latter case, $T(i)$ is treated as $A$'s output. 

When we discuss the complexity of $A$, we only count the number of group operation queries, denoted by \emph{the group operation complexity}, and allow it to make equality queries for free following~\cite{Mau05,Zhandry22a}. 
Assuming the zero-cost equality query makes our result stronger, and in fact describes the practice more appropriately. We refer to a more detailed discussion in~\cite[Remark 3.1]{Zhandry22a}.

\paragraph{Quantum generic group model.}
We extend the GGM to define the \emph{quantum} generic group model (QGGM). 
A generic algorithm $A$ in the QGGM
works over a working register $\mathbf{W}$, a query register $\mathbf{Q}$, and a table register $\mathbf{T}$.  
The registers  $\mathbf{W}$ and  $\mathbf{Q}$ are initialized to be $\ket{0...0}$. 
The register $\mathbf{T}$ stores $s$ group elements of $\gZ_N$ for some positive integers $s,N$. 
At the beginning, $\mathbf{T}$ is initialized as $\ket{y_1,...,y_{m},0,...,0}_{\mathbf{T}}$ 
where 
$(y_1,...,y_{m})\in \gZ_N^m$ is the input of $A$.  
$A$ can apply arbitrary quantum operations on $\mathbf{W}$ and $\mathbf{Q}$, but it can only act on $\mathbf{T}$ through the following types of queries: 
\begin{itemize}
    \item \emph{Group operation queries.} 
   Apply the following unitary $O_{\mathbf{Q},\mathbf{T}}$ on  $\mathbf{Q},\mathbf{T}$:
    \begin{equation}\label{eqn: group_operation_query}
    \ket{b,i,j}_{\mathbf{Q}} \otimes \ket{...,x_i,...,x_j,...}_\mathbf{T} \mapsto
    \ket{b,i,j}_{\mathbf{Q}} \otimes \ket{...,x_i + (-1)^b x_j,...,x_j,...}_{\mathbf{T}}
    \end{equation}
    if $i\neq j$ and otherwise it does nothing.
    \item \emph{Equality queries.} Apply the following operation on  $\mathbf{Q}$ and  $\mathbf{T}$:
    \[
        \ket{b,i,j}_{\mathbf{Q}} \otimes \ket{...,x_i,...,x_j,...}_{\mathbf{T}} \mapsto\ket{b\oplus t,i,j}_{\mathbf{Q}}\ket{...,x_i,...,x_j,...}_{\mathbf{T}}
    \]
    where $t=1$ if $x_i=x_j$ and $t=0$ otherwise.
\end{itemize}
Finally, $A$ outputs a classical string or a special symbol $\mathsf{group}$ along with an integer $i\in [s]$. In the latter case, 
$\mathbf{T}$ is measured and the $i$-th element in the measurement outcome 
is treated as the output of $A$. 

As in the classical GGM, the \emph{group operation complexity} of $A$ is defined by the number of group operation queries, and the equality queries are considered as free.

% \ifnum\llncs=0
% \begin{remark}
%     The QGGM has several differences from the classical GGM besides allowing quantum queries. First, the group operation query takes two indices $(i,j)$ instead of three indices $(i,j,k)$.
%     This modification is made because we cannot "overwrite" the $k$-th element in the quantum setting since that is a non-unitary operation.
%     Second, we put an upper bound $s$ for the size of the table. 
%     This is to capture the size of quantum memory available for the adversary. 
    
%     We remark that in many theorems (\Cref{thm: simulation_basic,thm: simulation_hybrid}), the memory size does not appear in the bounds, which means that these theorems give lower bounds for arbitrarily large quantum memory size.  
%     On the other hand, in \Cref{thm: simulation_memory} where we consider hybrid quantum-classical generic algorithms with memory restrictions, the bound depends on some relevant parameters.
% \end{remark}
% \begin{remark}\label{rem:eq_Shor}
% As we will see later in \Cref{sec: algo}, Shor's algorithm can be written as a generic algorithm in the QGGM.  
% Interestingly, Shor's algorithm does \emph{not} make use of the equality queries. 
% \end{remark}
% \fi

\paragraph{Parallel-query generic algorithms.}
We define \emph{parallel-query} generic algorithms in the QGGM.
A parallel-query generic algorithm $A$ in the QGGM works similarly to that in the QGGM except that it has $K$ query registers $\mathbf{Q}_1,...,\mathbf{Q}_K$ for some positive integer $K$ (referred to as the \emph{query width}) and is allowed to make parallel queries as follows:
\begin{itemize}
    \item \emph{Parallel group operation queries.} 
    Let $O_{\mathbf{Q}_k,\mathbf{T}}$ be a unitary that works as in \Cref{eqn: group_operation_query} where $\mathbf{Q}_k$ plays the role of $\mathbf{Q}$. 
    Then 
   apply the following operation on  $\mathbf{Q}_1,...,\mathbf{Q}_K$ and  $\mathbf{T}$:
    \[
        \bigotimes_{k\in [K]} \ket{b_k,i_k,j_k}_{\mathbf{Q}_k} \otimes \ket{x_1,...,x_s}_\mathbf{T} \mapsto
        \prod_{k\in [K]} O_{\mathbf{Q}_k,\mathbf{T}} \bigotimes_{k\in [K]} \ket{b_k,i_k,j_k}_{\mathbf{Q}_k} \otimes \ket{x_1,...,x_s}_\mathbf{T} 
    \]
    if $i_k\notin \{i_{k'}\}_{k'\in [K]\setminus \{k\}}\cup \{j_{k'}\}_{k'\in [K]}$ for all $k\in[K]$
    and otherwise it does nothing.\footnote{Intuitively, this condition means that multiple queries should not write to the same register and if one of the queries writes to some register, then that register should not be used as a control register for another query.
    Note that this does not prohibit parallel queries that share the same control register.}
    \item \emph{Parallel equality queries.} Apply the following operation on  $\mathbf{Q}$ and  $\mathbf{T}$:
    \[
         \bigotimes_{k\in [K]} \ket{b_k,i_k,j_k}_{\mathbf{Q}_k} \otimes \ket{x_1,...,x_s}_\mathbf{T} \mapsto\bigotimes_{k\in [K]} \ket{b_k\oplus t_k ,i_k,j_k}_{\mathbf{Q}_k} \otimes \ket{x_1,...,x_s}_\mathbf{T} 
    \]
    where $t_k=1$ if $x_{i_k}=x_{j_k}$ and $t_k=0$ otherwise.
\end{itemize}
We call the number of parallel group operation queries by the \emph{group operation depth}.
%of $A$.

\ifnum\llncs=0
\begin{remark}
    We do not consider parallel queries that mix group operation and equality queries because such queries can be split into a parallel group operation query and a parallel equality query. 
\end{remark}
\begin{remark}
    Strictly speaking, here we are extending the QGGM to deal with parallel queries; when $K$ is fixed to $1$, this model becomes the QGGM above. On the other hand, if we always measure the query register whenever the algorithm makes a query, it is not hard to see that the QGGM is equivalent to the classical GGM when we allow $s$ to be arbitrarily large.
\end{remark}
\fi

\paragraph{Convention.}
We often analyze the problems defined for a \emph{multiplicative} cyclic group $\cG$ in the (Q)GGM with $N=|\cG|$. 
In this case, we occasionally identify $x\in \gZ_N$ and $g^x\in \cG$ where $g$ is a generator of $\cG$.  
In particular, the generator $1\in \gZ_N$ is identified with $g\in \cG$ and the zero element $0\in \gZ_N$ is identified with $1\in \cG$. 
We also often abuse notation to write $\cG$ to mean the generic group oracle in the (Q)GGM. 
For example, a (parallel-query) generic algorithm for the DL problem is written as $A^\cG(g, g^x)$.

\subsection{Group-theoretic Problems}\label{subsec: Problems_in_QGGM}
\ifnum\llncs=0
\subsubsection{Problems} 
\fi
For a finite set $S$, we write $x \gets S$ to denote that an element $x$ is uniformly sampled from $S$ at random.
\paragraph{The Discrete Logarithm (DL) Problem.} 
In the discrete logarithm problem, an element $x \gets \Z_N$ is uniformly chosen at random. The first and second elements of the table $T$ are initialized by $T(1) = 1$ and $T(2)=x$ so that the input to the algorithm is $(g,g^x)$.
The adversary is asked to output $x$.
The advantage of the DL adversary $A^\cG$ is defined as follows:
\[
\adv_{\sf DL}(A^\cG)=\Pr_x\left[ A^{\cG }(g,g^x)\rightarrow x
\right],
\]
where the input $g,g^x$ denotes the elements stored in the table $T$.

\paragraph{The Computational/Decisional Diffie-Hellman Problem.}
In the computational Diffie-Hellman problem (CDH), two elements $x,y \gets \Z_N$ are randomly chosen. An instance $(g,g^x,g^y)$ is given to the adversary as elements in $T$. The adversary is asked to compute $g^{xy}$ in the table.
The advantage of the CDH adversary $A^\cG$ is defined as follows:
\[
\adv_{\sf CDH}(A^\cG)=\Pr_{x,y}\left[ A^{\cG }(g,g^x,g^y)\rightarrow g^{xy}
\right].
\]

In the decisional Diffie-Hellman problem (DDH), three random elements $x,y,r\gets \Z_N$ are chosen, and either $(g,g^x,g^y,g^{r})$ or $(g,g^x,g^y,g^{xy})$ is given to the adversary, as elements in $T$, and the adversary is asked to decide which is the case by outputting a decision bit $b\in\{0,1\}$.

The advantage of the DDH adversary $A^\cG$ is defined as follows:
\[
\adv_{\sf DDH}(A^\cG)=\left | \Pr_{x,y}\left[ A^{\cG }(g,g^x,g^y,g^{xy})\rightarrow 1 
\right] - \Pr_{x,y,r}\left[ A^{\cG }(g,g^x,g^y,g^{r})\rightarrow 1 
\right] \right | .
\]

\paragraph{The Multiple Discrete Logarithm Problem.}
In the $m$-multiple discrete logarithm problem ($m$-MDL), $m$ elements $x_1,...,x_m \gets \Z_N$ are independently and uniformly chosen at random. The $m$-MDL problem instance $g^{x_1},...,g^{x_m}$ is given to the adversary, stored in the $2,...,(m+1)$-th elements of $T$ along with the first element $g$ of $T$. The adversary is asked to find all of $x_1,...,x_m$.
The advantage of the $m$-MDL adversary $A^\cG$ is defined as follows:
\[
\adv_{m-\sf MDL}(A^\cG)=\Pr_{x_1,...,x_m}\left[ A^{\cG }(g,g^{x_1},...,g^{x_m})\rightarrow (x_1,...,x_m)
\right].
\]

\ifnum\llncs=0
\subsubsection{Classical lower bounds}
The following theorems state the lower bounds of the above problems in GGM.
\else
The GGM lower bounds of the above problems are as follows.
\fi
\ifnum\llncs=0
\begin{theorem}[GGM lower bound of DL/CDH/DDH~\cite{Mau05,Shoup97}]\label{thm: GGM_DL}
Let $Q$ be a positive integer, and $\cG$ be a prime-order cyclic group.
For a generic algorithm $A^\cG$ in GGM with $Q$ queries and for any $* \in \{{\sf DL,CDH,DDH} \}$, it holds that
\[
\adv_{*} (A^\cG) = O \left(\frac {Q^2} {|\cG|}\right).
\]
In particular, any constant-advantage algorithm in the GGM solving the DL/CDH/DDH problem makes at least $\Omega(\sqrt {|\cG|})$ queries.
\end{theorem}
\begin{theorem}[GGM lower bound of MDL~\cite{Yun15}]\label{thm: GGM_MDL}
Let $Q$ be a positive integer, and $\cG$ be a prime-order cyclic group.
For a generic algorithm $A^\cG$ in GGM with $Q$ queries for the $m$-MDL problem $\cG$, it holds that
\[
\adv_{m\text{-}\sf MDL} (A^\cG) = O \left(\left(\frac {e(Q+m+1)^2} {2m|\cG|}\right)^m\right).
\]
In particular, any constant-advantage algorithm in the GGM solving the $m$-MDL problem makes at least $\Omega(\sqrt {m|\cG|})$ group operations.
\end{theorem}
\else
\begin{theorem}[\cite{Mau05,Shoup97}]\label{thm: GGM_DL}
Let $Q$ be a positive integer, and $\cG$ be a prime-order cyclic group.
For a generic algorithm $A^\cG$ in GGM with $Q$ group operations and for any $* \in \{{\sf DL,CDH,DDH} \}$, it holds that
\[
\adv_{*} (A^\cG) = O \left(\frac {Q^2} {|\cG|}\right).
\]
In particular, any constant-advantage DL/CDH/DDH algorithm in the GGM makes at least $\Omega(\sqrt {|\cG|})$ queries.
\end{theorem}
\begin{theorem}[\cite{Yun15}]\label{thm: GGM_MDL}
Let $Q$ be a positive integer, and $\cG$ be a prime-order cyclic group.
For a generic algorithm $A^\cG$ in GGM with $Q$ group operations for the $m$-MDL problem $\cG$, it holds that
\[
\adv_{m\text{-}\sf MDL} (A^\cG) = O \left(\left(\frac {e(Q+m+1)^2} {2m|\cG|}\right)^m\right).
\]
In particular, any constant-advantage algorithm in the GGM solving the $m$-MDL problem makes at least $\Omega(\sqrt {m|\cG|})$ group operations.
\end{theorem}

\fi

\begin{remark}\label{remark: quantum_alg_in_GGM}
    We stress that the query complexity is the only complexity measure when showing the lower bounds. In particular, the lower bounds in~\Cref{thm: GGM_DL,thm: GGM_MDL} do apply for adversaries even with \emph{quantum} or \emph{unbounded} computational powers, as long as they only make $T$ classical queries. This observation is essential for our result.
\end{remark}

\section{Quantum Algorithms in the QGGM}\label{sec: algo}
This section presents generic quantum algorithms for the DL and MDL problems. 
Readers mainly interested in the lower bound can safely skip this section.

We first review Shor's algorithm for the DL problem with a closer look at the group operation complexity and its modification with classical preprocessing. The new MDL algorithm is presented at the end of this section.

We stress that all of these are generic. Also, they follow the standard approach for the abelian hidden subgroup problem (HSP) but the way to compute the relevant function defining HSP is different. Therefore, the correctness analysis follows from the known analysis, and we omit them (see, e.g.,~\cite{Lom04}).

Let $N$ be a positive integer and define $w_N:= \exp(2\pi i/N).$ The quantum Fourier transform $\QFT$ and its inverse $\QFT^\dagger$ are defined as follows:
\[
\QFT:\ket{x} \mapsto \frac{1}{\sqrt N}\sum_{k=0}^{N-1} w_N^{xk} \ket k,\text{ and }
\QFT^\dagger:\ket{k} \mapsto \frac{1}{\sqrt N}\sum_{x=0}^{N-1} w_N^{-xk}\ket x.\]

\subsection{The Discrete Logarithm Problem}
We rephrase the quantum algorithm due to Shor~\cite{Shor94} in detail and describe its variations in the hybrid setting and the depth-efficient version.

Let $\cG$ be a cyclic group of order $N$. 
Suppose that a generator $g$ and a handle $g^x$ representing the problem instance are given to the adversary for random $x$. The QGGM algorithm below finds $x$, where the below description roughly includes the square-and-multiply method to compute $g^{a+bx}$, which are omitted in the usual descriptions.
Let $|\cG|=N$, $n=\lceil \log_2 N\rceil.$ 

\subsubsection{The quantum DL algorithm}
The DL algorithm $A$ proceeds as follows.
\begin{enumerate}
    \item Given a problem instance $(g,g^x)$, the algorithm prepares the group elements of the form $g^{2^i}$ and $g^{2^j\cdot x}$, or the set
    \[
    D_x=\{g,g^2,g^{2^2},...,g^{2^{n-1}}\}\cup \{g^x,g^{2 x},g^{2^2 x},...,g^{2^{n-1} x}\}
    \]
    in the table using classical group operations. It prepares a quantum state $\ket{0,0}_A \otimes \ket{1,D_x}_T.$
    \item Applying $\QFT\otimes\QFT$ on the working register of $A$ to obtain 
    \[
    \sum_{a,b=0}^{N-1}\frac{\ket {a,b}_A}{{N}}\otimes\ket{1,D_x}.
    \]
    \item Using the binary expression of $a,b$ and $D_x$, the algorithm computes
    \[
    \sum_{a,b=0}^{N-1}\frac{\ket{a,b}_A\ket{g^{a+bx},D_x}_T}{{N}},
    \]
    applies $\QFT^\dagger \otimes \QFT^\dagger$, and measures the register $A$. For the measurement outcome $(u,v)\neq(0,0)$, return $v/u$ as an answer. Otherwise return $\bot$.
\end{enumerate}

The number of oracle queries is $O(n)=O(\log |\cG|)$; the construction of $D_x$ requires $O(n)$ queries, and computing $g^{a+bx}$ requires $2n$ queries, each of which is the controlled group operation multiplying $g^{2^i}$ or $g^{2^jx}$ on the first entry of the table.

Note that the quantum registers of this algorithm are essentially the working register $A$ and the first register of the table holding $g^{a+bx}.$ Furthermore, the quantum group operation accesses only one register of the remaining parts.
The following folklore theorem summarizes the result of this algorithm, regarding this observation.
\begin{theorem}\label{thm: plain_Shor}
    Let $\cG$ be a cyclic group. There exists a QGGM algorithm making $O(\log|\cG|)$ group operations that solves the discrete logarithm problem with an overwhelming probability. 
    
    This algorithm requires a quantum register holding $2n$-qubit\footnote{This can be reduced using the tricks in, e.g.,~\cite{Kitaev96,ME98}.} and a single group element, and classical storage holding $2n$ group elements. This algorithm does not require quantum access to classical storage.\footnote{Precisely speaking, it requires a QRACM storing a single group element, which is unavoidable when accessing group elements in classical memory.}
\end{theorem}

\subsubsection{Hybrid quantum-classical algorithms}
In the above algorithm, the construction of $D_x$ is entirely classical, and only the computation of $g^{a+bx}$ uses quantum power.
Since classical computation is much cheaper than quantum computing, it is tempting to reduce the later query complexity at the cost of the former classical preprocessing. 

We can modify the above algorithm taking this consideration into account, by exploiting the base-$p$ numeral system for $p\ge 2.$ Let $n_p:= \lceil \log_p N \rceil.$
This modified $p$-base algorithm uses the following set
\[
    D_x^{(p)}=\bigcup_{1\le k < p, 0 \le i <n_p}\left \{ g^{k\cdot p^i}\right\} \cup \bigcup_{1\le \ell < p, 0 \le j <n_p}\left\{g^{\ell \cdot p^j  x}\right\}
\]
instead of $D_x$ above. The only differences are the step for preparing $D_x^{(p)}$ and the way to compute $g^{a+bx}$. The preparation requires $O(p\log |\cG|/\log p)$ group operations, and the $p$-base exponent takes $O(\log |\cG|/\log p)$ quantum group operations. 
It is worth noting that this algorithm requires the QRACM access to a size-$p$ subset of $D^{(p)}_x$. 
The result of this hybrid algorithm is summarized as follows.

\begin{theorem}\label{thm: base-p_Shor}
    Let $\cG$ be a cyclic group and let $p>1$ be an integer. 
    There exists a generic hybrid algorithm with $O(\log |\cG|/\log p)$ quantum group operations and $O(p\log |\cG|/\log p)$ classical group operations that solves the DL problem with an overwhelming probability. 
    
    This algorithm requires a quantum register holding $2n$-qubit and a single group element, and classical storage holding $O(p\log |\cG|/\log p)$ group elements.
    This algorithm requires QRACM access to the classical $O(p)$ group elements simultaneously.
\end{theorem}

\subsubsection{Depth-efficient algorithms}\label{subsubsec: depth_efficient}
As suggested in~\cite{CW00}, we can exploit the binary tree with $2\log |\cG|$ leaves and $\log (2\log |\cG|)$ depth to reduce the depth of the DL algorithm. Precisely, we prepare the elements in $D_x$ as leaves of the binary tree and compute $g^{a+bx}$ using this tree with $2\log |\cG|-1$ internal nodes. Note that this computation is coherently done over internal nodes.

Furthermore, we can combine the binary tree idea with the base-$p$ numeric system with $D_{x}^{(p)}$. In this case, we prepare $g^{(a_{2i}+a_{2i+1} \cdot p) p^{2i}  }$ and $g^{(b_{2i}+b_{2i+1} \cdot p) p^{2i}  }$ for indices $a,b$ as leaf nodes by coherently multiplying the elements in $D_x^{(p)}.$ The tree has $ O(\log |\cG|/\log p)$ nodes and $\log\log |\cG| - \log\log p +O(1)$ depth.
This gives the following depth-efficient DL algorithm.
\begin{theorem}\label{thm: depth_Shor}
    Let $\cG$ be a cyclic group and let $p$ be an integer. There exists a generic quantum algorithm with $O(\log |\cG|/\log p)$ group operations of depth $\log\log |\cG|-\log\log p+O(1)$ that solves the discrete logarithm problem with an overwhelming probability. 

    It requires a quantum register storing $2n$-qubit and $O(\log |\cG|/\log p)$ group elements, and classical storage holding $O(p\log |\cG|/\log p)$ group elements. It also requires QRACM access to the classical $O(p)$ group elements simultaneously.
\end{theorem}

\subsection{The Multiple Discrete Logarithm Problem}
This section describes our new quantum MDL algorithm, where the adversary is given $m$ group elements $y_1=g^{x_1},...,y_m=g^{x_m}$ as inputs. The adversary's goal is to find $x_1,...,x_m$ using group operation queries.

The proposed algorithm follows the standard approach to hidden subgroup problems for the target function $f:\mathbb Z_N^{m+1} \rightarrow \cG$ given by
\begin{equation}\label{eqn: target MDL}
    f(k_0,...,k_m) = g^{k_0} y_1^{k_1} \cdots y_m^{k_m},
\end{equation}
which hides a rank-$m$ subgroup $H \le \mathbb Z_N^{m+1}$ generated by $\{x_ie_0 - e_i\}_{1\le i \le m}$; that is, it holds that $f(g_1)=f(g_2)$ for $g_1H=g_2H$. We stress the correctness follows from the previous analysis of the abelian HSP algorithms.
The improvement comes from the multi-exponentiation algorithm, which we recall below.

\subsubsection{The multi-exponentiation problem}
In the multi-exponentiation problem, we are given the elements $1,h_1,...,h_m$ and the nonnegative exponents $e_1, ..., e_m$, and asked to find $h_1^{e_1}\cdot ... \cdot h_m^{e_m}$ only using the multiplication. 
Pippenger~\cite{Pip80} showed the following result, which is known to be almost optimal.
\begin{proposition}[{\cite{Pip80}}]\label{prop:pip} 
Let $B$ be an integer, and $\lg m/ \lg B = o(1)$. Suppose $e_i \le B$ for all $i$. 
Given inputs $1,h_1,...,h_m$ and $e_1,...,e_m$, 
there is an efficient deterministic algorithm to compute $h_1^{e_1} \cdot ... \cdot h_m^{e_m}$ with 
\[
\lg B + \frac{(1+o(1)) m \lg B}{\lg (m \lg B)}
\]
multiplications.
\end{proposition}

\subsubsection{A multiple discrete logarithm algorithm}
Let $\cG$ be a cyclic group with order $N$. Suppose that we are given $g,y_1=g^{x_1},...,y_m=g^{x_m}$ as inputs.
As in the standard algorithm for HSP, the algorithm prepares a superposition 
    \[
    \bigotimes_{i=0}^{m} \left(\sum_{0\le k_i <N}\frac{\ket{k_i}}{\sqrt M} \right) = \sum_{0 \le k_0,...,k_m < N} \frac{\ket{k_0,...,k_m}}{\sqrt{M^{m+1}}}
    \]
using QFT and then compute the target function $f$ in~\Cref{eqn: target MDL} coherently using~\Cref{prop:pip}.

To execute~\Cref{prop:pip}, the condition $\lg m/\lg |\cG| = o(1)$ must hold. We also note that this algorithm requires large quantum memory. The result in this section is as follows.
\begin{theorem}\label{thm: MDL_alg}
    Let $\cG$ be an cyclic group and $m$ be a positive integer such that $\lg m/\lg |\cG| = o(1)$. 
    There exists a QGGM algorithm that solves the $m$-MDL problem using
    \[
    2\log |\cG| + \frac{(2+o(1))m \lg |\cG|} {\lg (m\lg |\cG|)}
    \]
    quantum group operation with an overwhelming probability.     
    If $m=\Omega(\log |\cG|),$ the amortized group operation complexity is $O(\log |\cG|/\log m)$ per DL instance.
\end{theorem}
% \fi
\han{On the hybrid? I guess it may not be helpful but there's no strong reason.}
\section{Quantum Lower Bounds in the QGGM}\label{sec: quantum}
In this section, we prove the quantum lower bounds of the DL and related problems in the QGGM. Our main technical tool is the following simulation theorem.
\begin{theorem}\label{thm: simulation_basic}
    Let $\cG$ be a group.  
    Suppose that a generic algorithm $A^\cG$ in the QGGM is given $m$ group elements as input and makes at most $q$ group operations.
    Then there exists a generic algorithm $B^\cG$ in the classical GGM for $\cG$, given the same inputs, which makes at most $2^{(q+1)m}$ group operations to the oracle such that the output distributions of $B^\cG(y)$ and $A^\cG(y)$ are identical for any input $y$.

    Furthermore, if $A^\cG$ has the group operation depth at most $d$, the corresponding algorithm $B$ exists with at most $2^{(d+1)m}$ group operations such that the output distributions are identical.
\end{theorem}

The generic algorithm $B$ may perform quantum or unbounded computation, but the group operations are all done classically. As observed in~\Cref{remark: quantum_alg_in_GGM}, the GGM lower bounds apply to the algorithm $B$.
With this observation in mind, the quantum lower bound of the DL problem in the QGGM is an immediate corollary of~\Cref{thm: simulation_basic}.
\begin{theorem}[Formal version of~\Cref{thm: intro_DL}]\label{thm: QGGM_DL}
    Let $\cG$ be a prime-order cyclic group.
    Any constant-advantage generic quantum algorithm solving the DL problem makes group operations of depth at least $\Omega(\log {|\cG|})$.
    
    More precisely, the following holds. 
    Let $d$ be a positive integer. For a QGGM algorithm $A^\cG$ making quantum group operations of depth $d$ for the DL problem over $\cG$, it holds that
    \[
        \adv_{\sf DL} (A^\cG) = O \left(\frac {2^{4d}} {|\cG|}\right).
    \]
\end{theorem}
\begin{proof}
    A DL instance consists of two group elements $(g,g^x)$, thus $m=2$ for the DL problem.
    For $m=2$, any generic quantum algorithm $A^\cG$ with a $d$ query depth can be simulated by a generic classical algorithm $B^\cG$ with $Q=2^{2(d+1)}$ queries by~\Cref{thm: simulation_basic}. The advantage of $B^\cG$ is bounded by $O(Q^2/|\cG|)$ due to~\Cref{thm: GGM_DL}, which shows the desired result.
    \ifnum\llncs=1
    \qed
    \fi
\end{proof}

The quantum CDH/DDH lower bound can be proven similarly. 
\begin{theorem}\label{thm: QGGM_CDH}
    Let $\cG$ be a prime-order cyclic group.
    Any constant-advantage generic quantum algorithm solving the CDH/DDH problem makes group operations of depth at least $\Omega(\log {|\cG|})$.
    
    More precisely, the following holds. 
    Let $d$ be a positive integer.
    For a generic quantum algorithm $A^\cG$ making quantum group operations of depth $d$ and for any $* \in \{{\sf CDH,DDH}\}$, it holds that
    \[
    \adv_{*} (A^\cG) = O \left(\frac {2^{8d}} {|\cG|}\right).
    \]
\end{theorem}

\subsection{Proof of~\Cref{thm: simulation_basic}}\label{subsec: proof simul basic}
We return to the proof of the main theorem. 
The idea of the proof is to \emph{simulate} the generic algorithm in the QGGM by using classical group operations in the GGM. 
The simulation algorithm exhaustively computes all branches of the original algorithm. 
This may take an unbounded time, but the group operation complexity is bounded, and only exponentially larger than that of $A$, which suffices for our purpose; again, the classical GGM lower bounds only consider the group operation complexity.

\begin{proof}[{\ifnum\llncs=0 Proof \fi }of~\Cref{thm: simulation_basic}]
Let $A$ be a parallel-query generic algorithm in the QGGM with group operation depth $d$ and query width $K$ that takes $m$ group elements $y_1,...,y_{m}$ as input. Without loss of generality, we assume that $y_1,...,y_{m}$ are not $0$.
Let $N=|\cG|$.
Then we construct a generic algorithm $B$ in the GGM that simulates $A$ as follows.

\paragraph{Initialization.} 
For the simulation, the algorithm $B$ makes use of a ``labeling function''
\[L:\mathbb Z_N[Y_1,...,Y_{m}]\rightarrow \mathbb [N] \cup \{\bot\},\] 
which is gradually updated during the simulation.\footnote{$L$ will take non-$\bot$ values only on polynomials of degree at most $1$ throughout the simulation.} Here, $\mathbb Z_N[Y_1,...,Y_{m}]$ denotes the $m$-variate polynomial ring over $\mathbb Z_N$ with indeterminates $Y_1,...,Y_{m}$. 
Also, we note that $[N]\cup \{\bot\}$ denotes the set of labels, and we will define the binary operation $\pm$ between labels, which should be distinguished from the addition/subtraction of integers.

Intuitively, when we have $L(f)=\ell\neq \bot$, it should be understood that we give a ``label'' $\ell\in \mathbb [N]$ to $f(y_1,...,y_{m})\in \mathbb Z_N$. 
For this to be well-defined, we always make sure that $L(f)=L(g)$ if and only if $f(y_1,...,y_{m})=g(y_1,...,y_{m})$ for any $f,g$ on which $L$ is defined (i.e., takes a non-$\bot$ value). 

The algorithm $B$ initializes $L$ as follows: 
\begin{enumerate}
\item Set $L(0)\leftarrow [N]$.\footnote{Recall this means $L(0)$ is sampled from $[N]$ uniformly at random.}
\item For $i=1,...,m$, uniformly set $L(Y_i)\in [N]$ under the constraint that 
$L(Y_i)\neq L(0)$ and 
$L(Y_i)=L(Y_j)$ if and only if $y_i=y_j$. 
Note that $B$ can do this because it can check if $y_i=y_j$ by making an equality query to the classical group oracle. 
\item Set $L(f):=\bot$ for all $f\notin  \{0,Y_1,...,Y_{m}\}$.\footnote{There are infinitely many elements in $\mathbb Z_N[Y_1,...,Y_{m}]$, but $B$ does not explicitly record that $L$ is defined to be $\bot$ on those inputs. The value of $L$ is understood to be $\bot$ unless it is explicitly defined to be a non-$\bot$ value.} 
\end{enumerate}
$B$ also defines a set $S\subseteq  \mathbb Z_N[Y_1,...,Y_{m}]$ of polynomials on which the value of $L$ is defined (i.e., not $\bot$). That is, $S:=\{0\}\cup \{Y_i\}_{i\in \{1,...,m\}}$. The set $S$ will be updated along with $L$ to ensure that it is always the set consisting of polynomials on which the value of $L$ is defined.

$B$ creates the following state as a simulation of the initial state for $A$:
\[
\ket{0...0}_{\mathbf{W},\mathbf{Q}_1,...,\mathbf{Q}_K}\otimes \ket{L(Y_1),...,L(Y_{m}),L(0),...,L(0)}_{\mathbf{T}}.
\]
% \han{we may need to say that we overwrite the notation of $\mathbf{W,Q,T}$ because those notations are already used by $A$, and $B$ may construct its own registers.}
% \takashi{I think it's reasonable to use the same notation for $A$'s registers and $B$'s simulation for them even without any explanation.}
During the simulation, we keep the invariance that for any $\ell\in \mathbb Z_n$ that appears in the register $T$ of any branch with non-zero amplitude, there is $f \in S$ such that $\ell=L(f)$. 
This is satisfied at this point since $S=\{0\}\cup \{Y_i\}_{i\in \{1,...,m\}}$. 

\paragraph{Local operation.}
When $A$ applies local operation on its registers $\mathbf{W},\mathbf{Q}_1,...,\mathbf{Q}_K$, $B$ also applies the same operation.

\paragraph{Parallel group operation query.}
Suppose $A$ makes a parallel group operation query. 
Let $S_{\mathsf{pre}}:=S$. (We introduce $S_{\mathsf{pre}}$ to record the set $S$ at the point of making the query since we will update the set $S$ during the simulation below.) 
Then $B$ does the following. Informally, the first step updates $L$ and $S$ to include the group elements potentially appearing after the query, and the second step defines the group operations over the labels. $B$ simulates the parallel group operation of $A$ in the last step.
\begin{enumerate}
    \item For each pair $(f,g)\in S_{\mathsf{pre}}^2$  (in arbitrary order), do the following:
    \begin{enumerate}
        \item Check if there is any $h\in S$ such that 
        \[h(y_1,...,y_{m})=f(y_1,...,y_{m})+ g(y_1,...,y_{m}).\]
        Note that $B$ can check this by making one group operation query for the RHS and many equality queries to the classical group oracle because $f(y_1,...,y_{m})$,  $g(y_1,...,y_{m})$, and $h(y_1,...,y_{m})$ have been already generated in the table of the classical group oracle.
        \begin{itemize}
            \item If there exists such  $h$, it sets
            \[L(f+ g):=L(h).\]
            Note that this is well-defined since the RHS does not depend on the choice of $h$.\footnote{As already mentioned, we always ensure that 
            $L(h)=L(h')$ if and only if $h(y_1,...,y_{m})=h'(y_1,...,y_{m})$ for all $h,h'\in S$. } 
            \item Otherwise, it uniformly sets $L(f+g)\leftarrow \mathbb [N]$ under the constraint that
            $L(f+g) \neq L(h)$ for all $h\in S$. 
        \end{itemize} 
        Then update $S\leftarrow S\cup \{f+g\}$. 
        \item Similarly define $L(f-g)$ and update $S\leftarrow S\cup \{f-g\}$.
    \end{enumerate}
    \item For labels $(\ell,\ell') \in \mathbb [N]^2$, define $\ell\pm \ell'$ as follows:\footnote{This is the binary operation between two labels, NOT the sum/subtraction of $\ell$ and $\ell'$ as integers or modulus $N$.}
    \begin{itemize}
        \item Check if there is  $(f,g)\in S_{\mathsf{pre}}^2$ such that $\ell=L(f)$ and $\ell'=L(g)$.
        \begin{itemize}
            \item If they exist, define  
            \[\ell\pm \ell':=L(f\pm g).\]
            Note that this is well-defined since the RHS does not depend on the choice of $(f,g)$.\footnote{Suppose that $L(f)=L(f')$ and $L(g)=L(g')$.
            Then on input $(y_1,...,y_{m})$, $f,f'$ and $g,g'$ have the same image, respectively, which implies that $f\pm g$ and $f'\pm g$ also have the same image under that input. 
            Thus, $L(f\pm g)=L(f'\pm g')$.
            }
            \item Otherwise, define $\ell \pm \ell':= \bot$. 
        \end{itemize}
    \end{itemize}
    \item Then $B$ simulates the group operation oracle by using the above defined operations for labels. 
    That is, it does the following:
    For each $k\in [K]$, let $\widetilde{O}_{\mathbf{Q}_k,\mathbf{T}}$ be the unitary that works as follows:
    \ifnum\llncs=1 \small\fi
    \begin{equation*}
        \ket{b_k,i_k,j_k}_{\mathbf{Q}_k} \ket{...,\ell_{i_k},...,\ell_{j_k},...}_\mathbf{T} \mapsto
        \ket{b_k,i_k,j_k}_{\mathbf{Q}_k} \ket{...,\ell_{i_k} + (-1)^{b_k} \ell_{j_k},...,\ell_{j_k},...}_{\mathbf{T}}
    \end{equation*}
    \ifnum\llncs=1 \normalsize\fi
    if $i_k\neq j_k$ and otherwise it does nothing. 
    Then apply the following operation on  $\mathbf{Q}_1,...,\mathbf{Q}_K$ and  $\mathbf{T}$:
    \ifnum\llncs=1 \small\fi
    \begin{equation*}
        \bigotimes_{k\in [K]} \ket{b_k,i_k,j_k}_{\mathbf{Q}_k} \otimes \ket{\ell_1,...,\ell_s}_\mathbf{T} \mapsto
        \prod_{k\in [K]} \widetilde{O}_{\mathbf{Q}_k,\mathbf{T}} \bigotimes_{k\in [K]} \ket{b_k,i_k,j_k}_{\mathbf{Q}_k} \otimes \ket{\ell_1,...,\ell_s}_\mathbf{T} 
    \end{equation*}
    \ifnum\llncs=1 \normalsize\fi
    if $i_k\notin \{i_{k'}\}_{k'\in [K]\setminus \{k\}}\cup \{j_{k'}\}_{k'\in [K]}$ for all $k\in[K]$ and otherwise it does nothing.
\end{enumerate}

\paragraph{Parallel equality query.}
When $A$ makes a parallel equality query, $B$ applies the following unitary:
 \[
         \bigotimes_{k\in [K]} \ket{b_k,i_k,j_k}_{\mathbf{Q}_k} \otimes \ket{\ell_1,...,\ell_s}_\mathbf{T} \mapsto\bigotimes_{k\in [K]} \ket{b_k\oplus t_k ,i_k,j_k}_{\mathbf{Q}_k} \otimes \ket{\ell_1,...,\ell_s}_\mathbf{T} 
    \]
    where $t_k=1$ if $\ell_{i_k}=\ell_{j_k}$ and $t_k=0$ otherwise.

\paragraph{Finalization.}
If $A$ outputs a classical string, $B$ outputs the same string. 
If $A$ outputs the special symbol $\mathsf{group}$ and an integer $i$, $B$ measures $\mathbf{T}$. 
Let $\ell\in [N]$ be the $i$-th element in the measurement outcome.
$B$ finds $f\in S$ such that $L(f)=\ell$. 
Then $B$ finds the index $i'$ such that the $i'$-th element stores $f(y_1,...,y_{m})$ in the table kept by its own classical group oracle. 
(Such $f$ and $i'$ must exist by the definition of the simulation.) 
Then $B$ outputs $\mathsf{group}$ and $i'$. 

\smallskip
The above completes the description of $B$. It is easy to see that $B$ perfectly simulates $A$. 
To see this, we consider a hybrid simulator $B'$ that has an additional capability to directly see $y_1,...,y_{m}$ and works as follows: $B'$ simulates the group operation oracle for $A$ as in the QGGM except that it first randomly chooses a random bijection $I:\mathbb Z_N \rightarrow \mathbb [N]$ and records $I(z)$ instead of $z$ in $T$ for any group element $z$. When $A$ makes a group operation query, $B'$ applies $I^{-1}$ to the relevant entries, applies the group operation, and then applies $I$ again. It is easy to see that $B'$ perfectly simulates $A$ because the random bijection $I$ just induces a basis change in $T$. Moreover, the ways of simulation by $B$ and $B'$ are perfectly indistinguishable from the view of $A$ because we can regard $B$ as doing the same as $B'$ except that it samples the bijection $I$ via lazy sampling through $L$.   

We count the number of group operation queries made by $B$. To do so, we observe that $B$ only generates group elements that can be generated by depth-$d$ applications of the group operation. In particular, the group elements generated after the first query are of the form $z_1 y_1 + ... + z_{m} y_{m}$ for $|z_i|\le 1$, because quantum group operations for $i=j$ is ignored by definition. Inductively, we can show that it only needs to generate group elements $z$ of the form
\[
    z=z_1 y_1 + ... + z_{m} y_{m}
\]
where $|z_j|\le 2^{d-1}$ for all $j=1, ..., m$. These group elements can be generated in the table of the classical group oracle by making $(2^{d}+1)^m \le 2^{m(d+1)}$ group operation queries.

    Combining the above arguments, we conclude that the classical GGM algorithm $B$ can perfectly simulate the algorithm $A$ by making at most $2^{(d+1)m}$ classical group operation queries.
    \ifnum\llncs=1 \qed \fi
\end{proof}

% \yun{Just to make sure: when $A$ is a quantum adversary, the algorithm $B$ can do quantum computation, but access the group oracle \emph{classically and reversibly}, right?  I guess $B$ will never overwrite previous nonzero entries.}
% \han{I think yes. $B$ always writes the output of queries in a new register.}
\section{Hybrid Quantum-Classical Algorithms}\label{sec: hybrid}
% This section presents the lower bounds of generic hybrid quantum-classical algorithms for group-theoretic problems.
% We model a hybrid quantum-classical algorithm by a sequence of \emph{quantum subroutines} followed by a classical post-processing algorithm. Each quantum subroutine can make an arbitrary number of classical group operation queries, but the number or the depth (as a parallel algorithm) of quantum group operations is bounded. If the number of the quantum group operations is bounded by $q$, we denote it by the $q$-query quantum subroutine, and the $d$-depth quantum subroutine is defined in a similar way. The quantum subroutines must measure the whole registers, including the table register $\mathbf T$, at the end of its execution.
% The formal model of the generic hybrid algorithms is given in the full version.

This section presents the lower bounds of generic hybrid quantum-classical algorithms for group-theoretic problems.
In~\Cref{subsec: hybrid_model}, we formalize the model of generic hybrid algorithms. 
\ifnum\llncs=1
The lower bounds and the hybrid simulation theorems are presented in~\Cref{subsec: hybrid_lowerbounds}. The proofs for the hybrid simulation theorem 
\ifnum\noappendix=0
are deferred to~\Cref{subsec: hybrid_simulation}.
\else
can be found in the full version.
\fi
\else
The lower bounds are presented in~\Cref{subsec: hybrid_lowerbounds} using the hybrid simulation theorem, which is proved in~\Cref{subsec: hybrid_simulation}.
\fi

\subsection{The Model of Hybrid Algorithms}\label{subsec: hybrid_model}
We establish the model of generic hybrid algorithms in this section. 
First, we define a \emph{classical group operation} $O^C_{\mathbf{Q,T}}$ in the QGGM, which is illustrated in~\Cref{fig: classical operation}, as follows.
\begin{enumerate}
    \item Measure the query register $\mathbf Q$.
    \item If the measurement outcome is $b,i,j$, measure the $i$-th and $j$-th entries of the register $\mathbf T$.
    \item Apply $O_{\mathbf Q,\mathbf T}.$
\end{enumerate}
Intuitively, when we apply the classical group operations, all the relevant registers contain classical information. A non-classical group operation query is called \emph{quantum}.

\begin{remark}
The classical group operation in the QGGM differs from the group operation in the GGM because the group oracles in the two models have different interfaces. 
In this section, an algorithm $A$ (and its components to be described below) is always a generic algorithm 
in the QGGM, making classical or quantum group operation queries in the QGGM. 
On the other hand, an algorithm $B$ is always an algorithm in the GGM, making group operation queries in the GGM.
\end{remark}

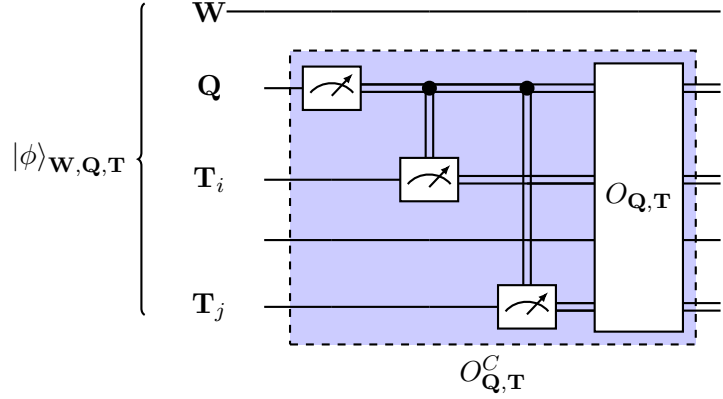
\begin{figure}[ht]
    \centering
    \begin{quantikz}
        \lstick[wires=5]{$\ket{\phi}_{\mathbf{W,Q,T}}$}
        &{\mathbf W}&\qw&\qw&\qw&\qw&\qw&\qw
        \\
        &{\mathbf Q}
        &&\meter{}
        \gategroup[wires=4,steps=4,style={dashed,fill=blue!20, inner sep=1pt},background,label style={label position=below,anchor=north,yshift=-0.2cm}]{$O^C_{\mathbf{Q,T}}$}
        &\cwbend{1}\cw&\cwbend{3}\cw&\gate[4,cwires={1,2,4}]{O_{\mathbf{Q,T}}}&\cw\\
        &{\mathbf T_i}&&\qw&\meter{}&\cw&\cw&\cw\\
        &&&\qw&\qw&\qw&\qw&\qw\\
        &{\mathbf T_j}&&\qw&\qw&\meter{}&\cw&\cw
        % &\gate[2,cwires={1,2}]{O_{\mathbf{Q,T}}}&\cw\\
        % &\meter{}&&\cw
    \end{quantikz}
    \caption{The classical group operation $O^C_{\mathbf{Q,T}}$: The single-line wires stand for quantum wires, while the double-line wires are for classical wires. 
    $\mathbf T_i, \mathbf T_j$ denote the $i$-th and $j$-th entries of $\mathbf T$.
    We assume that the measurement outcome of $\mathbf Q$ indicates the $i$-th and $j$-th entries in this diagram.
    Recall that $O_{\mathbf{Q,T}}$ is a group operation query.
    }
    \label{fig: classical operation}
\end{figure}

In our generic hybrid model of computation, a hybrid algorithm is defined by a generic algorithm in the QGGM that by itself performs only classical group operations, but has access to a \emph{quantum subroutine}, which is a generic quantum algorithm with a bounded number of quantum group operations but making an arbitrary number of classical group operations.

%In our generic hybrid model of computation, a hybrid algorithm is defined by a generic algorithm in the QGGM that has access to a \emph{quantum subroutine} during the computation, which is a generic quantum algorithm with a bounded number of quantum group operations but making an arbitrary number of classical group operations, and that only makes classical group operations. 

\begin{figure}[ht]
    \centering
    \footnotesize    
    \ifnum\llncs=1 \hspace*{-2cm} \fi
    \begin{quantikz}
        &\lstick{$\mathbf W$}&&\qw
        \gategroup[wires=5,steps=3,style={dashed,rounded corners,fill=blue!20, inner xsep=2pt},background]{$C_1$}
        &\gate[2]{V}&\qw
        &\qw
        & \gate[2]{V}
        \gategroup[wires=5,steps=4,style={dashed,rounded corners,fill=blue!20, inner xsep=2pt},background]{$C_2$}
        &\qw&\qw&\gate[2]{V}
        &\qw&\cdots&
        &\qw&\meter{}
        \\
        &\lstick{$\mathbf Q$}&&\gate[4]{O^C_{\mathbf Q,\mathbf T}}
        &\qw
        &\gate[4]{O^C_{\mathbf Q,\mathbf T}}
        &\gate[4]{O_1}
        &
        &\gate[4]{O^C_{\mathbf Q,\mathbf T}}
        &\gate[4]{O^C_{\mathbf Q,\mathbf T}}
        &
        &\qw&\cdots&
        &\gate[4]{O_q}&\meter{}
        \\
        &\lstick[wires=3]{$\mathbf T$}&&                           &\qw&
        &
        &\qw                           &&
        &\qw
        &\qw&\cdots&
        &&\meter{}
        \\
        &&&                        &\qw&
        &
        &\qw                           &&
        &\qw
        &\qw&\cdots&
        &&\meter{}
        \\
        &&&                        &\qw&
        &
        &\qw                           &&
        &\qw
        &\qw&\cdots&
        &&\meter{}
    \end{quantikz}
    \caption{The behavior of the quantum subroutine: $O_1,...,O_q$ denote the unitary operation that includes a single quantum group operation, and $C_0,....,C_{q-1}$ denote quantum algorithms that may include multiple classical group operations but no quantum group operations. 
    $V$ denotes an arbitrary quantum algorithm.
    All registers are measured after $O_q$ on a computational basis. 
    % \han{I am not sure how to draw a diagram properly. It may just confuse the readers. Is it okay?}\takashi{I think this figure is confusing because it omits the working register and then each $C_j$ looks like just one application of classical group operation. But in the actual definition, it can be arbitrary quantum computation with arbitrarily many classical group operations.}
    % \han{I will try to modify it later.}
    }
    \label{fig: quantum subroutine}
\end{figure}
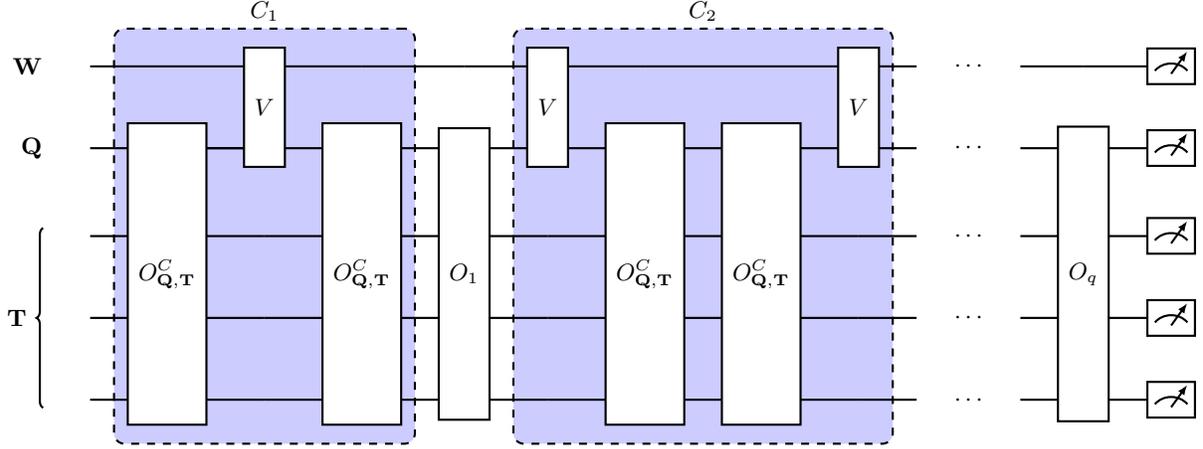

A quantum subroutine formalizes a quantum algorithm with a limited coherence time.
More precisely, we define a \emph{$q$-query quantum subroutine} by a generic quantum algorithm in the QGGM with at most $q$ \emph{quantum} group operations. After the $q$-th quantum group operation, it is forced to measure all registers on a computational basis. We call this measurement as the \emph{forced measurement}. On the other hand, the quantum subroutine can make an arbitrary number of \emph{classical} group operations. 
The total number of group operations is the summation of classical and quantum group operations. 
The quantum subroutine can perform an arbitrary quantum map over the registers $(\mathbf{W,Q})$ in the middle of its execution.
We illustrate the rough behavior of the table register in a quantum subroutine as in~\Cref{fig: quantum subroutine}.
As in the diagram, a quantum subroutine can be described by an alternating sequence of generic algorithms $(C_0,O_1,...,C_{q-1},O_q)$, where $C_j$ is a generic algorithm in the QGGM that may include multiple classical group operations and $O_j$ is a generic algorithm in the QGGM that includes a single quantum group operation for each $j$. Let $c_j$ be the number of classical group operations $C_j$ made.

Finally, a generic hybrid algorithm $A$ is specified by a tuple $(U_1,...,U_T,A_{T+1})$ of $T$ quantum subroutines and a follow-up generic algorithm that is connected as in~\Cref{fig: hybrid algorithm}. 
Here, $A_{T+1}$ is a generic algorithm that only makes classical group operations, and $U_j$ is a $q$-query quantum subroutine for each $j$.
Again, the generic hybrid algorithm can perform an arbitrary quantum map over $(\mathbf{Q,W})$ in the middle.
Recall that the quantum subroutine makes the forced measurements for all registers including $T$ on a computational basis, and the outcome of forced measurements will be given to the next subroutine (or $A_{T+1}$) as input. 

% the output of subroutines including the table $T$ \takashi{The usage of "output" seems inconsistent to our definition in Section 2.1.} must be measured in the computational basis. We illustrate the generic hybrid algorithm as in~\Cref{fig: hybrid algorithm}.

\begin{figure}[ht]
    \centering
    \small
    \ifnum\llncs=1 \hspace*{-2cm} \fi
    \begin{quantikz}
        \lstick{$\mathbf W$}&\gate[2]{V}&\qw&\gate[2]{V}&\qw&
        \qw&\cdots&&\qw&\qw
        &\gate[2]{V}&\qw&\gate[2]{V}&\qw\arrow[r]&
        \\
        \lstick{$\mathbf Q$}&&\gate[3]{U_1}\hphantom{wide}&\qw&\gate[3]{U_2}\hphantom{wide}&
        \qw&\cdots&&\qw&
        \gate[3]{U_T}\hphantom{wide}&&\gate[3]{A_{T+1}}\hphantom{wi}&&\qw\arrow[r]&
        \\
        \lstick[2]{$\mathbf T$}&\qw&&\qw&&
        \qw&\cdots&&\qw&
        &\qw&&\qw&\qw\arrow[r]&
        \\
        &\qw&&\qw&&
        \qw&\cdots&&\qw&
        &\qw&&\qw&\qw\arrow[r]&
    \end{quantikz}
    \caption{The generic hybrid algorithm with $T$ invocations of quantum subroutines: $U_1,...,U_T$ are quantum subroutines and include the measurement at the end. $A_{T+1}$ is a generic algorithm with classical group operations. $V$ denotes an arbitrary quantum algorithm.}
    \label{fig: hybrid algorithm}
\end{figure}
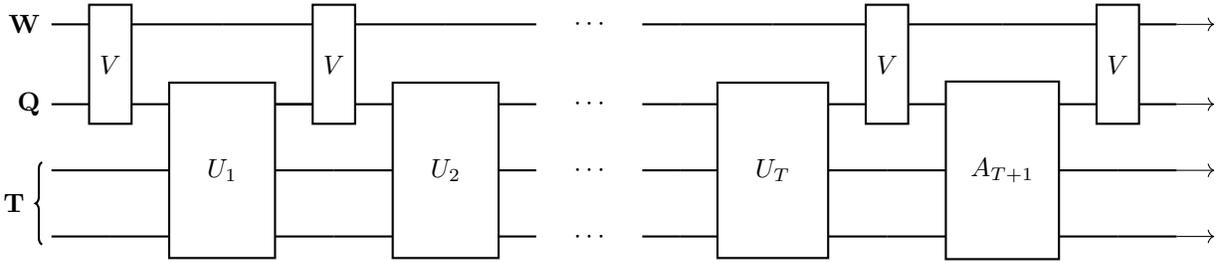

\paragraph{Parallel-query generic hybrid algorithms.}
We also define \emph{parallel-query} generic hybrid algorithm in the QGGM, by allowing quantum subroutines to be parallel-query generic algorithms in the QGGM. 
While a parallel classical group operation can be defined naturally, we only consider a classical group operation as defined above, which makes the simulation easier. Concretely, we separately use the parallel-query group operation oracle and the classical group operation oracle. The total number of queries is the summation of classical group operations and the query width $K$ 
% \takashi{It may be nice to remind what was $K$.} 
times the number of parallel-query quantum group operations.
% \footnote{We remark that the results below for parallel-query generic hybrid models are irrelevant to the total number of group operations. 
% Thus we do not intend to accurately define the total number of queries in this case, which is a bit complex.} 
% \takashi{But it seems that the total number of queries $Q$ appears in the bound?}
% \han{As I remember, $Q$ in the depth case could be a number of classical queries, but I just removed the footnote as it is confusing.}

A $d$-depth parallel quantum subroutine in the QGGM is defined by a parallel-query generic quantum algorithm in the QGGM such that the number of parallel quantum group operation queries is bounded by $d$. Again, the quantum subroutines can make intermediate classical queries, and can apply any quantum map on $(\mathbf{Q,W})$. Similarly with~\Cref{fig: quantum subroutine}, the $d$-depth parallel quantum subroutine can be described by a sequence of algorithms $(C_0,O_1,...,C_{d-1},O_d)$, where $C_j$ is a generic algorithm with $c_j$ classical group operations and $O_j$ is a generic quantum algorithm with a single parallel quantum group operation. The subroutine is forced to measure its all registers after $O_d$.

We characterize a parallel-query generic hybrid algorithm by a sequence of generic algorithms $(U_1,...,U_T,A_{T+1})$ where $U_j$ is a parallel-query generic quantum algorithm and $A_{T+1}$ is a generic algorithm with only classical group operations.

\begin{remark}
    We assume that the algorithm a priori fixes the sequence of oracle queries to classical and quantum group operations. This means that the algorithm cannot decide which oracle to call depending on its (classical) memory. As shown in~\cite{DFH22}, our result holds for the adaptive algorithm that chooses which oracle to query based on its memory with a slightly worse bound.
\end{remark}

\begin{remark}
This model of hybrid algorithms embraces a large class of hybrid algorithms considered in the literature as long as the subroutines have a bounded depth or number of quantum queries.
For example, both $d$-CQ and $d$-QC schemes in~\cite{CCL23} are included in our model. Even the algorithms in a higher hierarchy like ${\sf BPP}^{{\sf BQNC}^{\sf BPP}}$ (which is advocated in~\cite{ACCGSW23} as a proper model of hybrid algorithms) are described in this model, provided that the query number/depth bounds hold. 
In particular, the $d$-depth quantum subroutine can be interpreted as a $d$-QC scheme.
\end{remark}

% We define the generic hybrid algorithm as a generic algorithm in QGGM with oracle access to the generic quantum algorithm. Whenever the hybrid algorithm directly queries the group operation oracles, its input is measured on a computational basis right before the queries, while there is no such restriction on the generic quantum algorithm oracles. 
% Formally, the generic hybrid algorithm is represented by $A^\cG=(C_0^\cG, U_1^\cG,..., U_n^\cG, C_n^\cG)$, which resembles the CQ schemes in~\cite{CCL23}. The algorithm proceeds as follows:
% \[
% C_{0}^\cG \xrightarrow[]{c}\Pi \cdot U_{1}^\cG  \xrightarrow[]{c} \cdots \xrightarrow[]{c} C_{n-1}^\cG \xrightarrow[]{c}\Pi \cdot  U_{n}^\cG\xrightarrow[]{c} C_n^\cG,
% \]
% where $C_j$ is a generic algorithm such that it measures all group elements on a standard basis before each group operation query, $U_j$ is a generic quantum algorithm, and $\Pi$ denotes the standard basis measurement. The number of queries made by $C_j$ is denoted by $c_j$, and the query depth of $U_j$ is denoted by $q_j$. The transition $\xrightarrow[]{c}$ implies that the messages sent between algorithms are classical.

% In the context of generic algorithms, we assume that all algorithms $C_j, U_j$ are generic, and the measurements are also applied to the table register $T$.
% \han{Figure.}

\subsection{Lower Bounds for Hybrid Algorithms}\label{subsec: hybrid_lowerbounds}
This section presents two types of quantum query lower bounds of the DL problem against hybrid algorithms: The generic hybrid algorithm with $q$-query and $d$-depth quantum subroutines. 
We begin with the following simulation theorem for hybrid algorithms, 
\ifnum\llncs=1
\ifnum\noappendix=0
whose proofs are deferred to~\Cref{subsec: hybrid_simulation}.
\else
whose proofs can be found in the full version.
\fi
\else
whose proofs are deferred to the end of this section.
\fi
\begin{theorem}\label{thm: simulation_hybrid}
    Let $\cG$ be a group.
    Suppose that a generic hybrid algorithm $A^\cG$, taking $m$ group elements as inputs, makes at most $Q$ group operation queries (including both classical and quantum).
    \begin{itemize}
        \item If $A^\cG$ invokes $q$-query quantum subroutines $T$ times, then there exists a (randomized) classical GGM algorithm $B^\cG$ that perfectly simulates $A^\cG$ with 
        \[Q+T\cdot2^{q+1}(m+Q+1)^{2q}\]
        classical group operations.
        \item If $A^\cG$ invokes $d$-depth quantum subroutines $T$ times, then there exists a (randomized) classical GGM algorithm $B^\cG$ that perfectly simulates $A^\cG$ with 
        \[Q+T\cdot2^{2^d}(m+Q+1)^{2^d}\]
        classical group operations.
    \end{itemize}
\end{theorem}

As corollaries, we prove the lower bounds of the hybrid algorithm.
\begin{theorem}[Formal version of~\Cref{thm:intro_hybrid}]\label{thm: hybrid DL and more}
    Let $\cG$ be a prime-order cyclic group.
    Any constant-advantage generic hybrid algorithm solving the DL/CDH/DDH problems with $Q$ group operations, including both classical and quantum. If $Q=O(\poly\log |\cG|)$, it must make $\Omega(\log|\cG|/\log\log |\cG|)$ quantum queries of depth $\Omega(\log\log |\cG|)$ between some two consecutive forced measurements.

    More precisely, the following holds. Let $q\ge 1,d\ge 0$ be a positive integer. Let $*\in\{\sf{DL,CDH,DDH}\}$.
    \begin{itemize}
        \item 
        If a generic hybrid algorithm $A^\cG$ invokes $q$-query quantum subroutines $T$ times for $q\ge 1$, then it holds that 
    \[
    \adv_{*} (A^\cG) = O \left(\frac{\left( T\cdot2^{q+1}(m+Q+1)^{2q}\right)^2}{|\cG|}\right).
    \]
        \item If a generic hybrid algorithm $A^\cG$ invokes $d$-depth quantum subroutines $T$ times, then it holds that 
    \[
    \adv_{*} (A^\cG) = O \left(\frac{\left(T\cdot2^{2^d}(m+Q+1)^{2^d}\right)^2}{|\cG|}\right).
    \]
    \end{itemize}
\end{theorem}

\ifnum\llncs=0
\ifnum\llncs=1
\section{Proofs for Hybrid Lower Bounds}\label{subsec: hybrid_simulation}
\subsection{Simulation Theorems for Quantum Subroutines}
\else
\subsection{Proof of the Hybrid Simulation Theorem}\label{subsec: hybrid_simulation}
This section proves~\Cref{thm: simulation_hybrid}.
\fi
We first present the following variants of~\Cref{thm: simulation_basic} simulating the quantum subroutines.
\begin{lemma}[The simulation of $q$-query parallel quantum subroutine]\label{lem: simulation_subroutine_query}
    Let $\cG$ be a group.
    Suppose that $A^\cG$ is a $q$-query quantum subroutine in the QGGM, taking $m$ group elements and a classical string as inputs, making at most $Q$ group operation queries including both classical and quantum. Then there exists a generic algorithm $B^\cG$ in the GGM with 
        \[
        Q+2^{q+1}(m+Q+1)^{2q}
        \]
    classical group operations such that the output distribution of $B^\cG(y)$ and $A^\cG(y)$ are identical for any input $y$.
\end{lemma}
\begin{lemma}[The simulation of $d$-depth parallel quantum subroutine]\label{lem: simulation_subroutine_depth}
    Let $\cG$ be a group.
    Suppose that $A^\cG$ is a $d$-depth parallel quantum subroutine in the QGGM, taking $m$ group elements and a classical string as inputs, makes at most $Q$ classical group operation queries. Then there exists a generic algorithm $B$ in the GGM with
        \[
        2^{2^d}(m+Q+1)^{2^d}
        \]
    classical group operations such that the output distribution of $B^\cG(y)$ and $A^\cG(y)$ are identical for any input $y$. 
\end{lemma}

We can prove~\Cref{thm: simulation_hybrid} by invoking the above lemmas for each quantum subroutine.
\begin{proof}[{\ifnum\llncs=0 Proof \fi}of~\Cref{thm: simulation_hybrid}]
% \han{I am wondering if it is too sketch. Is it okay?}\takashi{I think this is fine (if my comment above is resolved.)}
    Let $A$ be a generic hybrid algorithm in the QGGM with $T$ invocations of quantum subroutines. Suppose that $A$ is characterized by $(U_1,...,U_T,A_{T+1})$ where $U_j$ is a $q$-query quantum subroutine (or $d$-depth quantum subroutine), and $A_{T+1}$ is a generic algorithm in the QGGM with classical group operations only.
    Let $c$ be the number of group operations $A_{T+1}$ made and $Q_j$ the total number of group operations $U_j$ made. 
    
    We construct an algorithm $B$ in the GGM that perfectly simulates $A$.
    We use the fact that the table register of the output of a quantum subroutine with $m$ group elements as input and $Q$ total queries has at most $m+Q$ nonzero group elements as each group operation makes at most a single new group element.

    As the first step, $B$ simulates $U_1$ using a sub-algorithm $B_1$.~\Cref{lem: simulation_subroutine_query,lem: simulation_subroutine_depth} assert that $B_1$ can simulate the $q$-query and $d$-depth quantum subroutine with
    \[
    Q_1+2^{q+1}(m+Q_1+1)^{2q},\text{ and }2^{2^d}(m+Q_1+1)^{2^d}
    \]
    group operations, respectively. 
    
    After the simulation, $B$ discards all but the output of the simulation, which is safely done as all registers are measured on a computational basis at the end of execution of $U_1$. Below, we only consider the non-discarded parts when we say the group elements or table register, etc. The discarded parts do not affect the remaining simulation.

    The table register has at most $m+Q_1$ nonzero group elements after $U_1$ since $Q_1$ group operations of $U_1$ add at most $Q_1$ new group elements to the table.
    $B$ then simulates $U_2$ using the non-discarded parts as input\footnote{Technically, we need slight variants of the above lemmas for the simulation, as the simulation input is described in the GGM, while the subroutines in the lemmas expect the group elements stored in the QGGM. Still, the procedures are identical, and we choose the modular analysis and omit the subtle details.} with a new sub-algorithm $B_2$. The complexity of the simulation is similar, and the result of measurement gives at most $m+Q_1+Q_2$ nonzero group elements in the table.

    Continuing this procedure until $U_T$. Finally, $B$ simulates $A_{T+1}$.
    Each classical group operation query of $A_{T+1}$ can be simulated by a single group operation of $B$, because the table register is measured after the execution of each quantum subroutine, and it is always classical during the execution of $A_{T+1}$. In other words, the simulation of $A_{T+1}$ requires $c$ classical group operations.

    We conclude that the algorithm $B$ in the GGM simulates the algorithm $A$ in the QGGM. The number of group operations is, if $A$ invokes $q$-query quantum subroutines,
    \ifnum\llncs=0
    \[
        \left(Q_1+2^{q+1}(m+Q_1+1)^{2q}\right)+\left(Q_2+2^{q+1}(m+Q_1+Q_2+1)^{2q}\right)+ ...+c
        \le Q+T\cdot 2^{q+1} (m+Q+1)^{2q},
    \]
    \else
    \begin{align*}
        &\left(Q_1+2^{q+1}(m+Q_1+1)^{2q}\right)+\left(Q_2+2^{q+1}(m+Q_1+Q_2+1)^{2q}\right)+ ...+c
        \\
        &\le Q+T\cdot 2^{q+1} (m+Q+1)^{2q},
    \end{align*}
    \fi
    where we use $Q_1+Q_2+...+Q_T+c\le Q,$  
    and
    \[
    2^{2^d}(m+Q_1+1)^{2^d}+ 2^{2^d}(m+Q_1+Q_2+1)^{2^d} + ... + c \le Q + T \cdot 2^{2^d}(m+Q+1)^{2^d}
    \]
    if $A$ invokes $d$-depth quantum subroutines, where we use $c\le Q$.
    \ifnum\llncs=1 \qed \fi
\end{proof}

\ifnum\llncs=0
It remains to prove~\Cref{lem: simulation_subroutine_query,lem: simulation_subroutine_depth}. 
\else
\subsection{Proofs of the Simulation Theorems for Quantum Subroutines}
\fi
We first prove the $d$-depth parallel quantum subroutine case, which can be proven similarly to~\Cref{thm: simulation_basic} with some modifications. 
The query number case needs a new idea regarding the branches and the modifications used in the depth case.

% \subsubsection{Proof of~\Cref{lem: simulation_subroutine_depth}}
We first prove~\Cref{lem: simulation_subroutine_depth}.
The basic idea of the proof is that, when $A$ makes a classical group operation, $B$ can update the label function $L$ and the set $S$ with a single group operation. 
This is because the algorithms measure the relevant registers. As many parts of the proof resemble one of~\Cref{thm: simulation_basic}, we highlighted the differences in the simulation in red.

\begin{proof}[{\ifnum\llncs=0 Proof \fi}of~\Cref{lem: simulation_subroutine_depth}]
    Let $A$ be a generic quantum subroutine in the QGGM with at most $d$ quantum group operation depth that takes $m$ group elements $y_1,...,y_{m}$ and \red{classical string $z$} as input. We assume that $y_1,...,y_{m}$ are not $0$, and let $N=|\cG|.$ 
    Suppose that $A$ is characterized by \red{a sequence of algorithms $C_0,O_1,...,C_{d-1},O_d$} with generic algorithm $C_j$ with $c_j$ classical group operations and generic algorithm $O_j$ with a single parallel quantum group operation.

    We construct a generic algorithm $B$ in the GGM that simulates $A$ by following the construction in~\Cref{thm: simulation_basic}, except for some specifications of group operation queries. In particular, the local operation, the equality query, and the finalization step are identical. 
    The initialization is only slightly different as the algorithm $A$ takes a classical string as a part of the input.
    We also apply the parallel group operation query procedure for the parallel quantum group operation queries, but we need a slightly different procedure for classical group operations. The correctness of the simulation procedure can be proven in the same way as the original proof.

    \paragraph{Initialization.} 
    The algorithm $B$ parses the input into the \red{classical string $z$} and group elements $y_1,...,y_m$ stored in the table register, and
    prepares a label function $L:\mathbb Z_N[Y_1,...,Y_{m}]\rightarrow \mathbb [N] \cup \{\bot\}$ and $S\subseteq  \mathbb Z_N[Y_1,...,Y_{m}]$.  Set $S:=\{0\}\cup \{Y_i\}_{i\in \{1,...,m\}}$ and initialize $L$ as in the previous proof. 
    $B$ creates the following state as a simulation of the initial state for $A$:
    \[
    \red{\ket{z}_{\mathbf{W},\mathbf{Q}}}\otimes \ket{L(Y_1),...,L(Y_{m}),L(0),...,L(0)}_{\mathbf{T}}
    \]
    
    Recall that $S\subset \mathbf Z_N[Y_1,...,Y_{m}]$ is the set of polynomials on which the value of $L$ is defined. 
    %at some point. \takashi{"at some point" is confusing to me. Can we remove this?}
    Before describing the classical group operation simulation, we make the following observation for the parallel quantum group operation query. Recall the simulation of parallel quantum group operation in~\Cref{subsec: proof simul basic} and $S$ is updated during the simulation.
    
    \begin{clm}\label{claim: hybrid_depth}
    Suppose that $S_{\sf pre}$ be a set $S$ right before a parallel quantum group operation and $S_{\sf post}$ be a set $S$ right after the same parallel quantum group operation. Then it holds that $|S_{\sf post}| \le 2|S_{\sf pre}|^2.$
    \end{clm}
    \begin{proof}
        In the parallel group operation query simulation, the set of new elements added in $S_{\sf post}$ is included in the set
        \[
        S':=\{h=f\pm g:f,g \in S_{\sf pre}\},
        \] 
        which has at most $2|S_{\sf pre}|^2$ different elements. 
        As $S_{\sf pre}$ includes $0$ by definition, $S_{\sf pre} \subset S'$ holds. This implies the claim.
        \ifnum\llncs=1 \qed \fi
    \end{proof}

    % \paragraph{Measurement.}\han{Do we need this description?}\takashi{I don't see why this is needed. Isn't that explained in Step 1 of the simulation for the classical group operation?}
    % As preparation for classical group operations, we describe the measurement procedures, which were implicit in the proof of~\Cref{thm: simulation_basic}.
    % In that proof, we proved that, for any intermediate quantum state of $A$, $B$ constructs the corresponding quantum state; the only difference is that $B$ stores the output of the label function $L$ in the table register while $A$ stores the group elements in the table register. This means that $A$'s measurements of any register can be perfectly simulated by $B$, by measuring the corresponding registers.\footnote{This is implicit in the original proof because it does not restrict the algorithm $B$ unitary, so $B$ can do measurement if needed.}

    \paragraph{\red{Classical group operation query.}} When $A$ makes a classical group operation query, $B$ does the following. Let $S_{\mathsf{pre}}:=S$.
    \begin{enumerate}
        \item Measure $B$'s query register $\mathbf Q$ to obtain $(b,i,j)$, and do nothing if $i=j$. Otherwise, measure the $i$-th and $j$-th entries of $B$'s table register to obtain the labels $\ell_i,\ell_j$.
        \item Find a pair $(f,g) \in S_{\mathsf{pre}}^2$ such that $\ell_i = L(f)$ and $\ell_j = L(g)$.
        % If they do not exist, define $L(f+(-1)^b g)=\bot.$
        Check if there is any $h\in S$ such that 
        \[
            h(y_1,...,y_{m})=f(y_1,...,y_{m})+(-1)^bg(y_1,...,y_{m}).
        \]
        % Note that such a pair must exist.
        %and the values of $f(y_1,...,y_m)$ and $g(y_1,...,y_m)$ do not depend on the choice of pairs, thus it requires a single group operation.
        \begin{itemize}
            \item If there is such $h$, it sets $L(f+(-1)^bg):=L(h)$.
            \item Otherwise, it uniformly sets $L(f+(-1)^bg)\leftarrow [N]$ under the constraint that $L(f+(-1)^bg)\neq L(h) $ for all $h \in S$.
        \end{itemize}
        Then update $S\leftarrow S\cup\{f+(-1)^b g\}.$ 
        \item For all pairs $(f',g') \in S_{\sf pre}^2$ such that $L(f')=\ell_i$ and $L(g')=\ell_j$, set $L(f'+(-1)^b g') := L(f+(-1)^b g)$ and update $S \leftarrow S\cup \{f'+(-1)^b g'\}.$
        \item Define $\ell_i +(-1)^b \ell_j := L(f+(-1)^b g).$
        \item Then, $B$ simulates the classical group operation oracle as follows:
        \begin{equation*}
        \ket{b,i,j}_{\mathbf{Q}} \otimes \ket{...,\ell_{i},...,\ell_{j},...}_\mathbf{T} \mapsto
        \ket{b,i,j}_{\mathbf{Q}} \otimes \ket{...,\ell_{i} + (-1)^{b} \ell_{j},...,\ell_{j},...}_{\mathbf{T}}
        \end{equation*}
        for $i\neq j$, and otherwise it does nothing. 
    \end{enumerate}
    Since the register $\mathbf Q$ and the $i$ and $j$-th entries of $\mathbf T$ are measured, this step only needs a single group operation for computing $f(y_1,...,y_{m})+(-1)^b g(y_1,...,y_{m})$.

    \paragraph{Group operation complexity.} 
    We count the number of classical group operations made by $B$. Let $S_{j-1}$ be the set $S$ right before the parallel quantum group operation in $O_j$ and let $S_d$ be the final set $S$. Recall $C_0$ makes $c_0$ classical group operations and the simulation of each classical group operation takes a single group operation by $B$. We have $|S_0| \le m +1+ c_0,$ where $m+1$ is the initial elements included in $S$. Also, using~\Cref{claim: hybrid_depth} and the fact that $C_j$ makes $c_j$ classical group operations, we have
    \[
    |S_{j}| \le 2|S_{j-1}|^2 + c_{j} \le 2(|S_{j-1}| + c_j)^2,
    \]
    and $|S_d| \le 2|S_{d-1}|^2$. From this, we inductively prove that
    \[
        |S_j| \le 2^{2^j-1}(m+1+c_0+...+c_{j-1})^{2^j}
    \]
    for all $j \le d$.
    Since $|S_d|$ is an upper bound of the number of group operations, 
    we conclude that $B$ can perfectly simulate $A$ with
    \[
        2^{2^d-1} \left( m +1+ c_0 + c_1 + ...+ c_{d-1} \right)^{2^d} \le 2^{2^d} \left( m + Q+1\right)^{2^d}
    \]
    group operations. This concludes the result for the $d$-depth parallel quantum subroutine.
    \ifnum\llncs=1 \qed \fi
\end{proof}

% \subsubsection{Proof of~\Cref{lem: simulation_subroutine_query}}
Now we move to the proof of~\Cref{lem: simulation_subroutine_query}.
The main idea of this case is to consider the branches. In the beginning, there is only a single branch with nonzero amplitude in the table register.
Two observations for proving~\Cref{lem: simulation_subroutine_query} are 1) for simulating a group operation over a fixed (classical) table in a single branch, we only need a tiny number of new group elements, and 2) for each group operation, the number of branches is multiplied by a bounded number; looking ahead, it is $2(m+Q+1)^2$. As the simulation procedure in~\Cref{thm: simulation_basic} does not consider the branch, we need some more work in this case.

\begin{proof}[{\ifnum\llncs=0 Proof \fi}of~\Cref{lem: simulation_subroutine_query}]
Let $A$ be a $q$-query quantum subroutine characterized by $(C_0,O_1,...,C_{q-1},O_q)$. Here, $C_j$ is a generic algorithm with $c_j$ classical group operations and $O_j$ is a generic algorithm with a single quantum group operation. We assume that $A$ takes $m$ group elements $y_1,...,y_{m}$ and {classical string $z$} as input and suppose that each $y_j$ is not $0$, and let $N=|\cG|.$
As in the above proof, we construct a generic algorithm $B$ in the GGM that simulates $A$. We assume that $A$ takes $(z,y_1,....,y_m)$ as input, where $z$ is a classical string and $y_1,...,y_m$ are group elements stored in the table register.

\paragraph{Initialization.} 
The algorithm $B$ parses the input into the {classical string $z$} and group elements $y_1,...,y_m$ stored in the table register.
$B$ prepares a label function $L:\mathbb Z_N[Y_1,...,Y_{m}]\rightarrow \mathbb [N] \cup \{\bot\}$ and $S\subseteq  \mathbb Z_N[Y_1,...,Y_{m}]$.  Set $S:=\{0\}\cup \{Y_i\}_{i\in \{1,...,m\}}$ and initialize $L$ as in the previous proof. 
% follows.
% \begin{enumerate}
% \item Set $L(0)\leftarrow [N]$.
% \item For $i=1,...,m$, uniformly set $L(Y_i)\in [N]$ under the constraint that 
% $L(Y_i)\neq L(0)$ and 
% $L(Y_i)=L(Y_j)$ if and only if $y_i=y_j$. 
% Note that $B$ can do this because it can check if $y_i=y_j$ by making an equality query to the classical group oracle. 
% \item Set $L(f):=\bot$ for all $f\notin  \{0,Y_1,...,Y_{m}\}$.
% \end{enumerate}
$B$ creates the following state as a simulation of the initial state for $A$:
\[
\ket{z}_{\mathbf{W},\mathbf{Q}}\otimes \ket{L(Y_1),...,L(Y_{m}),L(0),...,L(0)}_{\mathbf{T}}
\]

Additionally, $B$ initializes \red{a rooted tree structure $\mathcal T$} with a root $v_{0} = \{L(0)\}\cup \{L(Y_i)\}_{i\in \{1,...,m\}}$, without any further vertex. 
During the simulation, the tree $\mathcal T$ will be updated along with $S$ and $L$, so that \red{a leaf node represents the table register of a branch with a nonzero amplitude.} 
For example, the unique leaf node $v_{0}$ of the initial tree $\mathcal T$ includes all information of the table register of the unique initial branch.

% \paragraph{Local operation.}
% When $A$ applies some operation on its local registers $\mathbf{W},\mathbf{Q}$, $B$ also applies the same operation.

\paragraph{\red{Quantum group operation query.}}
Suppose $A$ makes a quantum group operation query.
Let $S_{\mathsf{pre}}:=S$ and $\mathcal T_{\mathsf{pre}}:=\mathcal T$ for recording the set $S$ and $\mathcal T$ at the point of making the query.
Then $B$ does the following: 
\begin{enumerate}
    \item For each \red{leaf node $v_{l}$ of $\mathcal T_{\mathsf{pre}}$} (in arbitrary order), and for each \red{$(\ell,\ell') \in v_l^2$}, do the following:
    \begin{enumerate}
        \item Find a pair $(f,g) \in S_{\sf pre}^2$ such that $L(f)=\ell$ and $L(g)=\ell'$. Check if there is any $h\in S$ such that 
            \[h(y_1,...,y_{m})=f(y_1,...,y_{m})+ g(y_1,...,y_{m}).\]
            \begin{itemize}
                \item If there exists such $h$, it sets
                \[L(f+ g):=L(h).\]
                % Note that this is well-defined since the RHS does not depend on the choice of $h$.
                \item Otherwise, it uniformly sets $L(f+g)\leftarrow \mathbb [N]$ under the constraint that
                $L(f+g) \neq L(h)$ for all $h\in S$. 
            \end{itemize} 
        % \end{itemize}
        \item Update $S\leftarrow S\cup \{f+g\}$. \red{Add $v_{l+1} := v_l \cup \{L(f+g)\}$ as a child node of $v_l$.}
        \item For all pairs $(f',g') \in S_{\sf pre}^2$ such that $L(f')=\ell$ and $L(g')=\ell'$, set $L(f'+g'):=L(f+g)$ and update $S\leftarrow S\cup \{f'+g'\}$.
    \end{enumerate}
    \item Similarly add $v_{l+1} :=v_l\cup \{L(f-g)\}$ as a child node of $v_l$ and update $S$ and $L$ properly.
    \item Define $\ell\pm \ell'$ for labels $(\ell,\ell') \in \mathbb [N]^2$ as follows:
    \begin{itemize}
    \item Check if there is  $(f,g)\in S_{\mathsf{pre}}^2$ such that $\ell=L(f)$ and $\ell'=L(g)$. 
        \begin{itemize}
            \item If they exist, define  
            \[\ell\pm \ell':=L(f\pm g).\]
            \item Otherwise, define $\ell \pm \ell':= \bot$. 
        \end{itemize}
    \end{itemize}
    \item Apply the unitary $\widetilde{O}_{\mathbf{Q},\mathbf{T}}$ that is defined by:
    \[
        \ket{b,i,j}_{\mathbf{Q}} \otimes \ket{...,\ell_{i},...,\ell_{j},...}_\mathbf{T} \mapsto
        \ket{b,i,j}_{\mathbf{Q}} \otimes \ket{...,\ell_{i} + (-1)^{b} \ell_{j},...,\ell_{j},...}_{\mathbf{T}}
    \]
    if $i\neq j$ and otherwise it does nothing.
\end{enumerate}

\paragraph{Classical group operation query.} 
When $A$ makes a classical group operation query, $B$ does the following. This is identical to the classical group operation in the $d$-depth subroutine case, except for the update of $\mathcal T$.
Precisely, in the second step, \red{$B$ updates  $v_l \leftarrow v_l \cup \{L(f+(-1)^b g)\}$ for each leaf node $v_l$}.
% \begin{enumerate}
%     \item Measure $B$'s query register $\mathbf Q$ to obtain $(b,i,j)$, and do nothing if $i=j$. Otherwise, measure the $i$-th and $j$-th entries of $B$'s table register to obtain the labels $\ell_i,\ell_j$.
%     \item Check if there is $(f,g) \in S$ such that $\ell_i = L(f)$ and $\ell_j = L(g)$.
%     If they do not exist, define $L(f+(-1)^b g)=\bot.$
%     If they exist, check if any $h\in S$ such that 
%     \[
%         h(y_1,...,y_{m})=f(y_1,...,y_{m})+(-1)^bg(y_1,...,y_{m}).
%     \]
%     \begin{itemize}
%         \item If there is such $h$, it sets $L(f+(-1)^bg):=L(h)$.
%         \item Otherwise, it uniformly sets $L(f+(-1)^bg)\leftarrow [N]$ under the constraint that $L(f+(-1)^bg)\neq L(h) $ for all $h \in S$.
%     \end{itemize}
%     Then update $S\leftarrow S\cup\{f+(-1)^b g\},$ and \red{update $v_l \leftarrow v_l \cup \{f+(-1)^b g\}$} for each leaf node $v_l.$
%     \item Then, $B$ simulates the classical group operation oracle as follows:
%     \begin{equation*}
%         \ket{b,i,j}_{\mathbf{Q}} \otimes \ket{...,\ell_{i},...,\ell_{j},...}_\mathbf{T} \mapsto
%         \ket{b,i,j}_{\mathbf{Q}} \otimes \ket{...,\ell_{i} + (-1)^{b} \ell_{j},...,\ell_{j},...}_{\mathbf{T}}
%     \end{equation*}
%     for $i\neq j$, and otherwise it does nothing. 
% \end{enumerate}

\paragraph{Finalization.}
Finalization is identical to the previous simulations.

\paragraph{Analysis.} 
The main difference of this simulation from the previous simulations is using $\mathcal T$. 
In particular, $B$ computes the labels for $f\pm g$ only when $f,g$ comes from a single leaf node.
We argue that it suffices for simulating the group operation queries. For an intermediate quantum state
\[
    \sum_{w,b,i,j,X} \alpha_{w,b,i,j,X}\ket{w}_{\mathbf W}\ket{b,i,j}_{\mathbf Q} \ket{X}_{\mathbf T},
\]
we say that $X$ is a \emph{nontrivial} table if $\alpha_{w,b,i,j,X}$ is nonzero for some $w,b,i,j$.
\begin{clm}\label{claim: nontrivial_table}
    Right before the group operation query, for a nontrivial table $X$, there is a leaf node $v_l$ of $\mathcal T$ such that the following holds: 
    For any group element included in $X$, there is $f \in v_l$ such that $f(y_1,...,y_{m})=x$. We say that $v_l$ corresponds to $X.$
\end{clm}
\begin{proof}
    We use induction. Before the first group operation query, the unique nontrivial table $X=(y_1,y_2,...,y_{m},0,...)$ corresponds to the initial leaf node $v_{0} = \{L(0)\}\cup \{L(Y_i)\}_{i\in \{1,...,m\}}$.
    Consider a group operation query, in which the statement holds right before the query. We prove that the statement still holds after the query.

    Suppose that the group operation is quantum.
    Let $X=(x_1,x_2,...)$ be a nontrivial table, and let $v_l$ be the corresponding leaf before the query at this point. 
    Consider a fixed branch
    \[
        \ket{w}_{\mathbf W}\ket{b,i,j}_{\mathbf Q} \ket{X}_{\mathbf T}.
    \]
    Since $v_l$ corresponds to $X$, there exist $f,g \in S$ such that $f(y_1,...,y_{m})=x_i$, $g(y_1,...,y_{m})=x_j$ and $L(f), L(g) \in v_l$ hold. The simulation appends $v_{l+1} = v_l \cup \{L(f+(-1)^b g)\}$ as a child node of $v_l.$ 
    Since the quantum group operation query only changes the $i$-th entry by $x_i$ to $x_i +(-1)^b x_j$, we can easily check the $v_{l+1}$ corresponds to the table after query.

    For a classical group operation query, a similar argument works with the same path $v_l$ since the simulation adds $L(f+(-1)^b g)$ to the set $v_l$.
    \ifnum\llncs=1 \qed \fi
\end{proof}

This claim ensures that each nontrivial table corresponds to some leaf node in the simulation so that the simulation works well as in the previous proofs. 

In the remainder of the proof, we count the number of group operations made by $B$. To do so, we need some calculations for leaf nodes. Let $v_l$ be a leaf node. $B$ appends the child nodes having one more group element that $v_l$ for each quantum group operation. 
Similarly, $B$ adds a single group element to the leaf nodes for each classical group operation. This implies that all nodes have at most $m+1+Q$ elements during the simulation.

After the initialization, there is only a single leaf node. For each quantum group operation, $B$ appends at most $2|v_l|^2$ child nodes to the leaf node $v_l$. Using $|v_l|\le m+Q+1,$ the final tree $\mathcal T$, after $q$ quantum group operations, has at most
\[
2^q(m+Q+1)^{2q}
\]
leaf nodes. Since each non-root node and each classical group operation requires a single group operation to update, the total number of group operations made by $B$ is bounded by
\[
Q + 2^q(m+Q+1)^{2q} + 2^{q-1}(m+Q+1)^{2(q-1)} + ... \le Q+2^{q+1} (m+Q+1)^{2q}.
\]
This concludes the proof.
\ifnum\llncs=1 \qed \fi
\end{proof}
\fi

\section{Memory-bounded Algorithms}\label{sec: memory}
In this section, we consider the generic algorithms with \emph{bounded memory}. 
More precisely, we consider algorithms that have two memories, classical and quantum, and that store a limited number of group elements in their quantum register. Furthermore, we assume that the algorithms can coherently access a few group elements stored in the classical memory in a single group operation. In other words, we assume that the algorithms only have a small amount of quantum random access classical memory (QRACM).

\ifnum\llncs=0
\ifnum\llncs=1
\section{Model and Proofs for Memory-bounded Algorithms}\label{sec: QRACM}
\fi
\subsection{Quantum and Classical Memory Models}\label{subsec: QRACM}
Recall that algorithms in (Q)GGM interact with a black-box table register $\mathbf T$ that stores group elements in $\Z_N$.
In this section, we assume that $\mathbf T$ holds two components $\mathbf T_C$ and $\mathbf T_Q = \cH_{\Z_N}^{\otimes t}$. 
The second component, $\mathbf T_Q$, is a quantum memory with a bounded size. Generic algorithms can coherently access or store group elements in $\mathbf T_Q$ in superposition.

The first component $\mathbf T_C$, corresponding to the classical memory, always stores group elements in a computational basis, i.e., classical group elements. We additionally assume that the quantum algorithm has a restriction on \emph{coherently} accessing $\mathbf T_C$. In other words, $\mathbf T_C$ is not a quantum random access (classical) memory (QRACM). We may assume that there is a QRACM holding a small number of group elements, which will be formally discussed below.

Recall that the register $\mathbf T$ is of the form $\mathbf T = \cH_{\Z_N}^{\otimes s}$. We decompose it into $\mathbf T_Q \otimes \mathbf T_C$ for $\mathbf T_Q = \cH_{\Z_N}^{\otimes t}$ where $s\ge t$.
We additionally assume that each $i$-th component of $\mathbf T$ for $i> t$ always holds a classical group element. 
The quantum group operation query
\[
    \ket{b,i,j}_{\mathbf Q} \otimes \ket{...,x_i,...,x_j,...}_{\mathbf T} \mapsto
    \ket{b,i,j}_{\mathbf Q} \otimes \ket{...,x_i + (-1)^b x_j,...,x_j,...}_{\mathbf T}
\]
has the following restrictions: In each query, the indices should be one of the following choices and obey the corresponding conditions.
\begin{enumerate}
    \item\label{item: qop} (Group operations for quantum registers) The second and third registers of $\mathbf Q$ hold indices (that may be in superpositions) indicating the group elements in the quantum register $\mathbf T_Q$; that is, $i,j \le t$ (in any branch) always holds.
    \item\label{item qcop} (Group operations for quantum-classical registers) The second register of $\mathbf Q$ holds indices (that may be in a superposition) indicating the group elements in $\mathbf T_Q$, and the third register of $\mathbf Q$ is classical (i.e., measured before query) and indicates a group element in $\mathbf T_C$; that is, $i\le t$ (in any branch) and $j>t$ holds.
    \item\label{item: cop} (Group operations for classical registers) The second and third registers of $\mathbf Q$ are both classical (i.e., measured in the computational basis before query), and the stored indices $i,j$ indicate group elements in $\mathbf T_C$; that is, $i,j>t$ holds.
\end{enumerate}

The first and second options are basically quantum (if the algorithm does not measure $\mathbf Q$), and the last option should be classical. 
Therefore, combining the first and second group operations as quantum group operations is convenient.

In general, we consider the case that a (small) QRACM is available. We model a QRACM as a storage containing $r$ group elements. The QRACM is only involved in quantum-classical group operations. 
When querying a quantum group operation, the algorithm must specify $r$ indices $j_1,...,j_r>t$ for the classical register. The QRACM loads those elements, and the group operation is made between QRACM and quantum register $\mathbf T_Q$ or just in $\mathbf T_Q$, after which the data in QRACM is discarded. 
Note that the second register must hold indices (in a superposition) indicating the group elements in $\mathbf T_Q$ since the result of the quantum group operation is written in that register.
Therefore, the quantum group operation can be described as follows, by omitting the QRACM's data loading and deletion.

\begin{itemize}
    \item (Group operations for quantum registers and QRACM) The algorithm specifies (not in superposition) a set $J=\{j_1,...,j_r\}$ of indices such that $j_k>t$ for all $k=1,...,r$. 
    The second register must indicate the group elements in $\mathbf T_Q$, possibly in a superposition. The third register holds indices $j_1,...,j_r$ or indicates the group elements in $\mathbf T_Q$. Then it applies $O_{\mathbf{Q,T}}$.
\end{itemize}

\fi

In summary, the memory-bound QGGM is similar to the generic hybrid quantum-classical algorithm in the QGGM with $q$-query quantum subroutines as described in~\Cref{sec: hybrid}, with the additional memory constraints described above. 
The group operations for classical registers {\ifnum\llncs=0(\Cref{item: cop}) \fi}are always considered as classical group operations. 
On the other hand, the quantum group operations always act on the quantum registers and QRACM.

\ifnum\llncs=0
\subsection{Lower Bounds for Memory-Bounded Algorithms}
\else
Due to the space constraints, we only present the simulation theorem for memory-bounded algorithms and the lower bounds in this model.
The formal model of quantum and classical memory and the proof for the simulation theorem 
\ifnum\noappendix=0
are given in~\Cref{sec: QRACM}.
\else
can be found in the full version.
\fi
\fi

We present the simulation theorem for memory-bounded algorithms first.

\begin{theorem}\label{thm: simulation_memory}
    Let $\cG$ be a group.
    Suppose that a generic hybrid algorithm $A^\cG$ makes $C$ classical group operations and invokes $q$-query quantum subroutines $T$ times. If the quantum memory of $A^\cG$ only can store $t$ group elements and $A^\cG$ can access QRACM of $r$ group elements, then there exists a (randomized) classical GGM algorithm $B^\cG$ that perfectly simulates $A^\cG$ with 
        \[C+2T\cdot(2t(t-1+r)+1)^{q}\]
    classical group operations.
\end{theorem}

Similarly to the other cases, this theorem directly implies the following lower bounds.
\begin{theorem}[Formal version of~\Cref{thm:intro_memory}]\label{thm: memory DL and more}
    Let $\cG$ be a prime-order cyclic group. Any constant-advantage generic hybrid algorithm solving the DL/CDH/DDH problems with quantum memory holding $t=O(1)$ group elements and QRACM of $r=O(1)$ group elements must make either $C=\Omega(\sqrt{|\cG|})$ classical group operations or $\Omega(\log |\cG|)$ quantum group operations.

    More precisely, the following holds. If a generic hybrid algorithm $A^\cG$ in the QGGM invokes $q$-query quantum subroutines $T$ times, it holds that for any $*\in\{\sf{DL,CDH,DDH}\}$
    \[
        \adv_{*} (A^\cG) = O \left(\frac{\left(C+2T\cdot (2t(t-1+r)+1)^q\right)^2}{|\cG|}\right).
    \]
\end{theorem}

\ifnum\llncs=0
\subsection{Proof of the Memory-Bounded Simulation Theorem}
We prove the following memory-bounded simulation theorem for the quantum subroutines. The proof of~\Cref{thm: simulation_memory} is almost identical to the hybrid case, except that it uses~\Cref{lem: simulation_subroutine_memory} that has a nice property that the number of classical queries $C$ is not involved in the exponential term. 

In the proof of~\Cref{lem: simulation_subroutine_memory}, the main difference is the contents of $\mathcal T$ where each node only includes $t$ elements, and identifies a potential branch in $\mathbf T_Q.$

\begin{lemma}\label{lem: simulation_subroutine_memory}
    Let $\cG$ be a group.
    Suppose that $A^\cG$ is a $q$-query quantum subroutine in the QGGM, taking group elements and a classical string as inputs, makes at most $C$ classical group operation. If the quantum memory of $A^\cG$ only can store $t$ group elements and $A^\cG$ can access QRACM of $r$ group elements, then there exists a generic algorithm $B^\cG$ in the GGM with 
        \[
        C+2\cdot (2t(t-1+r)+1)^{q}
        \]
    classical queries such that the output distribution of $B^\cG(y)$ and $A^\cG(y)$ are identical for any input $y$.
\end{lemma}

\begin{proof}
The proof is almost identical to that of~\Cref{lem: simulation_subroutine_query}, except that we change the contents of the rooted tree structure to only include group elements in quantum memory.
Let $A$ be a $q$-query quantum subroutine characterized by $(C_0,O_1,...,C_{q-1},O_q)$. Here, $C_j$ and $O_j$ are generic algorithms with $c_j$ classical group operations and a single quantum group operation, respectively. We assume that $A$ takes $m$ group elements $y_1,...,y_{m}$ and {classical string $z$} as input. Suppose that the input group elements are not $0$, and let $N=|\cG|.$
We will construct a generic algorithm $B$ in the GGM that simulates $A$.

\paragraph{Initialization.}
The algorithm $B$ prepares a label function $L:\mathbb Z_N[Y_1,...,Y_{m}]\rightarrow \mathbb [N] \cup \{\bot\}$ and $S\subseteq  \mathbb Z_N[Y_1,...,Y_{m}]$. Set $S:=\{0\}\cup \{Y_i\}_{i\in \{1,...,m\}}$ and initialize $L$ as in the previous proof. 
$B$ creates the following state as a simulation of the initial state for $A$:
\[
\ket{z}_{\mathbf{W},\mathbf{Q}}\otimes \ket{L(Y_1),...,L(Y_{m}),L(0),...,L(0)}_{\mathbf{T}}
\]

Additionally, $B$ initializes \red{a rooted tree $\mathcal T$ with a root $v_{0} =\{L(Y_i)\}_{i\in \{1,...,t\}}$ \emph{as a multiset}}, without any further vertex.\footnote{If $t>m$, we set $Y_i=0$ for $i>m$. If $t<m$, this implies that the initial quantum register does not include all input group elements.} \red{The nodes of $\mathcal T$ indicate potential branches of quantum memory $\mathbf T_Q$.} 
If a branch includes a group element multiple times, the corresponding multiset will include the same element with the same multiplicity.
% \han{I think set is more safe, as it works well for swap operations (while we do not mention it). I will change it accordingly.}
During the simulation, the tree $\mathcal T$ will be updated along with $S$ and $L$.

\paragraph{Local operation.}
When $A$ applies some operation on its local registers $\mathbf{W},\mathbf{Q}$, $B$ also applies the same operation.

\paragraph{\red{Quantum group operation query.}}
Suppose $A$ makes a quantum group operation query. \red{Note that $A$ specifies the QRACM index set $J=\{j_1,j_2,...,j_r\}$, which can be obtained by $B$ as well.}
Let $S_{\mathsf{pre}}:=S$ and $\mathcal T_{\mathsf{pre}}:=\mathcal T$ for recording the set $S$ and $\mathcal T$ at the point.
Then $B$ does the following: 
\begin{enumerate}
    \item For each leaf node {$v_l$ of $\mathcal T_{\mathsf{pre}}$} (in arbitrary order), do the following:
    \begin{enumerate}
        \item For each $(\ell,\ell')\in v_l^2$, do the following:
        \begin{enumerate}
            \item Find a pair $(f,g) \in S_{\sf pre}^2$ such that $L(f) = \ell$ and $L(g)=\ell'$. Check if there is any $h\in S$ such that 
            \[h(y_1,...,y_{m})=f(y_1,...,y_{m})+ g(y_1,...,y_{m}).\] 
            \begin{itemize}
                \item If there exists such $h$, it sets
                \[L(f+ g):=L(h).\]
                \item Otherwise, it uniformly sets $L(f+g)\leftarrow \mathbb [N]$ under the constraint that
                $L(f+g) \neq L(h)$ for all $h\in S$. 
            \end{itemize} 
            \item Update $S\leftarrow S\cup \{f+g\}$. \red{Define $v_{l+1}: = v_l \setminus \{L(f)\}\cup \{L(f+g)\}$.\footnote{If there is multiple $L(f)$ in $v_l$, it reduces the multiplicity of $f$ by 1.} Add $v_{l+1}$ as a child node of $v_l$.}
            \item For all pairs $(f',g') \in S_{\sf pre}^2$ such that $L(f')=\ell$ and $L(g')=\ell'$, set $L(f'+g'):=L(f+g)$ and update $S\leftarrow S\cup \{f'+g'\}$.
            \item Similarly add $v_{l+1} := v_l \setminus \{L(f)\}\cup \{L(f-g)\}$ as a child node of $v_l$ and update $S$ and $L$ properly.
        \end{enumerate}
        \item \red{For each $\ell \in v_l$ and $k\in [r]$, do the following:}
        \begin{enumerate}
            \item {Obtain the label $\ell'$ from the $j_k$-th entry of $B$'s table register, and find a pair $(f,g)\in S_{\sf pre}^2$ such that $\ell=L(f) $ and $\ell'=L(g)$.} Note that the $j_k$-th entry of $B$ is classical and $B$ needs not to measure it.
            \item Check if there is any $h\in S$ such that 
            \[h(y_1,...,y_{m})=f(y_1,...,y_{m})+ g(y_1,...,y_{m}).\]
            \begin{itemize}
                \item If there exists such $h$, it sets
                \[L(f+ g):=L(h).\]
                \item Otherwise, it uniformly sets $L(f+g)\leftarrow \mathbb [N]$ under the constraint that
                $L(f+g) \neq L(h)$ for all $h\in S$. 
            \end{itemize} 
            \item Update $S\leftarrow S\cup \{f+g\}$. \red{Define $v_{l+1} := v_l \setminus \{L(f)\}\cup \{L(f+g)\}$. Add $v_{l+1}$ as a child node of $v_l$.}
            \item For all pairs $(f',g') \in S_{\sf pre}^2$ such that $L(f')=\ell$ and $L(g')=\ell'$, set $L(f'+g'):=L(f+g)$ and update $S\leftarrow S\cup \{f'+g'\}$.
            \item Similarly add $v_{l+1} := v_l \setminus \{L(f)\}\cup \{L(f-g)\}$ as a child node of $v_l$ and update $S$ and $L$ properly.
        \end{enumerate}        
    \end{enumerate}
    \item Define $\ell\pm \ell'$ for labels $(\ell,\ell') \in \mathbb [N]^2$ as follows:
    \begin{itemize}
    \item Check if there is  $(f,g)\in S_{\mathsf{pre}}^2$ such that $\ell=L(f)$ and $\ell'=L(g)$. 
        \begin{itemize}
            \item If they exist, define  
            \[\ell\pm \ell':=L(f\pm g).\]
            \item Otherwise, define $\ell \pm \ell':= \bot$. 
        \end{itemize}
    \end{itemize}
    \item Apply the unitary $\widetilde{O}_{\mathbf{Q},\mathbf{T}}$ that is defined by:
    \[
        \ket{b,i,j}_{\mathbf{Q}} \otimes \ket{...,\ell_{i},...,\ell_{j},...}_\mathbf{T} \mapsto
        \ket{b,i,j}_{\mathbf{Q}} \otimes \ket{...,\ell_{i} + (-1)^{b} \ell_{j},...,\ell_{j},...}_{\mathbf{T}}
    \]
    if $i\neq j$ and otherwise it does nothing.
    \item \red{Finally, add a copy of $v_l$ as a child node of $v_l$.}\footnote{This is because some information can be erased from the nodes because of the memory bound.}
\end{enumerate}

\paragraph{\red{Classical group operation query.}} 
When $A$ makes a classical group operation query, $B$ does the following. 
% While we describe the procedure for the classical group operations, the procedure below readily applies for the group operations for classical registers.
Let $S_{\sf pre}: = S$ and $\mathcal T_{\mathsf{pre}}:=\mathcal T$ for recording the set $S$ and $\mathcal T$ at the point.
\begin{enumerate}
    \item Measure $B$'s query register $\mathbf Q$ to obtain $(b,i,j)$, and do nothing if $i=j$. Otherwise, measure the $i$-th and $j$-th entries of $B$'s table register to obtain the labels $\ell_i,\ell_j$.
    \item Find a pair $(f,g) \in S_{\sf pre}^2$ such that $\ell_i = L(f)$ and $\ell_j = L(g)$. Check if there is any $h\in S$ such that 
    \[
        h(y_1,...,y_{m})=f(y_1,...,y_{m})+(-1)^bg(y_1,...,y_{m}).
    \]
    \begin{itemize}
        \item If there is such $h$, it sets $L(f+(-1)^bg):=L(h)$.
        \item Otherwise, it uniformly sets $L(f+(-1)^bg)\leftarrow [N]$ under the constraint that $L(f+(-1)^bg)\neq L(h) $ for all $h \in S$.
    \end{itemize}
    Then update $S\leftarrow S\cup\{f+(-1)^b g\}.$
    \red{If $i\le t$, for each leaf node $v_l$ of $\mathcal T_{\sf pre}$ (in arbitrary order) such that $L(f) \in v_l$, update $v_l \leftarrow v_l \setminus \{L(f)\}\cup \{L(f+(-1)^bg)\}$.}
    \item For all pairs $(f',g') \in S_{\sf pre}^2$ such that $L(f')=\ell_i$ and $L(g')=\ell_j$, set $L(f'+(-1)^b g') := L(f+(-1)^b g)$ and update $S \leftarrow S\cup \{f'+(-1)^b g'\}.$
    \item Define $\ell_i +(-1)^b \ell_j := L(f+(-1)^b g).$
    \item Then, $B$ simulates the classical group operation oracle as follows:
    \begin{equation*}
    \ket{b,i,j}_{\mathbf{Q}} \otimes \ket{...,\ell_{i},...,\ell_{j},...}_\mathbf{T} \mapsto
    \ket{b,i,j}_{\mathbf{Q}} \otimes \ket{...,\ell_{i} + (-1)^{b} \ell_{j},...,\ell_{j},...}_{\mathbf{T}}
    \end{equation*}
    for $i\neq j$, and otherwise it does nothing. 
\end{enumerate}
% Since the register $\mathbf Q$ and the $i$-th and $j$-th entries of $\mathbf T$ are measured, this step only needs a single group operation for computing $f(y_1,...,y_{m})+(-1)^b g(y_1,...,y_{m})$.

% \paragraph{Classical group operation query.} 
% When $A$ makes a classical group operation query, $B$ does the following. This is identical to the classical group operation in the $d$-depth subroutine case, except for the update of $\mathcal T$.
% Precisely, in the second step, \red{$B$ updates  $v_l \leftarrow v_l \cup \{f+(-1)^b g\}$ for each leaf node $v_l$} when updating  $S\leftarrow S\cup\{f+(-1)^b g\}.$

\paragraph{Finalization.}
Finalization is identical to the previous simulations.

% \han{I didn't write an analysis yet.}
\paragraph{Analysis.} 
The contents in each node is the main difference from the proof of~\Cref{lem: simulation_subroutine_query}.
Recall that for an intermediate quantum state
\[
    \sum_{w,b,i,j,X} \alpha_{w,b,i,j,X}\ket{w}_{\mathbf W}\ket{b,i,j}_{\mathbf Q} \ket{X}_{\mathbf T},
\]
we say that $X$ is a nontrivial table if $\alpha_{w,b,i,j,X}$ is nonzero. We prove the following variant of~\Cref{claim: nontrivial_table}, which ensures that the algorithm $B$ correctly simulates $A$.
\begin{clm}
    Right before the group operation query, for a nontrivial table $X$, there is a leaf node $v_l$ of $\mathcal T$ such that the following holds: 
    Let $x_1,...,x_t$ be the first $t$ elements of $X$. There is a leaf node $v_l = \{\ell_1,...,\ell_t\}$ in $\mathcal T$ such that there exist $(f_i)_{i \in [t]}$ where 
    $f_i(y_1,...,y_{m})=x_i$ and $L(f_i)= \ell_i$ for all $i\in [t]$. We say that $v_l$ corresponds to $X.$
\end{clm}
\begin{proof}
    We use induction. Initially, $X=(y_1,y_2,...,y_{m},...)$ is the unique nontrivial table corresponds to the initial leaf node $v_{0} =\{L(Y_1),...,L(Y_{t})\}$. Consider a group operation query, in which the statement holds right before the query. We prove that the statement still holds after the query.

    For the quantum group operation,
    let $X=(x_1,x_2,...)$ be a nontrivial table and $v_l=\{\ell_1,...,\ell_t\}$ be the corresponding leaf node. Let $(f_i)_{i \in [t]}$ be functions such that $f_i(y_1,...,y_{m})=x_i$ and $L(f_i)= \ell_i$ for all $i\in [t]$.
    Consider a fixed branch
    \[
        \ket{w}_{\mathbf W}\ket{b,i,j}_{\mathbf Q} \ket{X}_{\mathbf T}.
    \]
    It holds that $f_i(Y_1,...,Y_m)=x_i$ for all $i\in [t]$.
    If $j>t$, let $f_j \in S$ be such that $f_j(Y_1,...,Y_m)=x_j$.
    After the query, it is easy to see that $v_{l+1}=v_l \setminus \{L(f_i)\}\cup\{L(f_i+(-1)^bf_j)\}$ defined above corresponds to the new table register. This argument also holds for the classical group operation with $i\le t$. For the group operations over classical registers, $v_l$ still corresponds to $X$.
    For a nontrivial branch unaffected by any group operation corresponds to the copy of $v_l$.
    \ifnum\llncs=1 \qed \fi
\end{proof}

We count the number of group operations made by $B$. 
For each quantum group operation, $B$ appends one child node when it makes one group operation for simulating $A$.
On the other hand, each classical group operation of $A$ requires a single group operation of $B$, and it may alter some contents of nodes but do not create any new node.
Thus it suffices to count the number of nodes in $\mathcal T$ for computing the query complexity of $B$.

Recall that there are two options for quantum group operations. The following arguments show that each quantum query appends at most $2t(t-1+r)+1$ child nodes for each leaf node.
\begin{itemize}
    \item For the operations for quantum registers, at most $2t(t-1)$ different branches can appear, where factor 2 represents the choice of $\pm$ and $t(t-1)$ is for the choice of indices in the quantum register. 
    %This shows that a quantum query can be simulated by $2t(t-1)$ classical group operations per branch.
    \item For the quantum-classical group operations with QRACM, the first index can be one of $t$ elements in $\mathbf T_Q$ and the second index is one among $|J|=r$ group elements. Therefore, each query introduces at most $2tr$ new branches.
    % \item For the classical group operation (\Cref{item: cop}), the random simulation procedure described in~\Cref{lem: simulation_subroutine} allows us to sample the measurement outcome of indices $i,j$, and only a single classical group operation suffices for this simulation. Also, note that the number of nontrivial branches is not changed by this operation.
\end{itemize}
The final $+1$ is for the branches unaffected by group operations.
This implies that after $q$ quantum queries, there are at most
\[
\sum_{i=0}^{q} (2t(t-1+r)+1)^i\le 2 \cdot\left(2t(t-1+r)+1\right)^q
\]
different branches in $\mathcal T$ at the end, where we used $t,r\ge 1$. As a single classical group operation can be simulated by a single group operation of $B$, the total number of group operations in the GGM made by $B$ is bounded by $C+2 \cdot\left(2t(t-1+r)+1\right)^q$, which concludes the proof.
\ifnum\llncs=1 \qed \fi
\end{proof}
\fi

\paragraph{Acknowledgment.} 
We thank Oded Regev and Martin Eker{\aa} for the insightful discussion.
We also would like to thank anonymous reviewers of CRYPTO 24.
Minki Hhan was partially supported by a KIAS Individual Grant QP08980, and the National Research Foundation of Korea (NRF) grant funded by the Korea government (MSIT) (No.~2021R1A2C1010690). Aaram Yun was supported by the National Research Foundation of Korea (NRF) grant funded by the Korea government (MSIT) (No.~2021R1A2C1010690).

\bibliographystyle{labelalpha}
\bibliography{ref}

\ifnum\noappendix=0
\newpage
\ifnum\llncs=1
\appendix

\fi
\fi

\end{document}